%% file: CCM_JMLR.tex
\documentclass[twoside,11pt]{article}

\usepackage{jmlr2e}
\usepackage{xcite}

\usepackage{xr}
\makeatletter
\newcommand*{\addFileDependency}[1]{% argument=file name and extension
  \typeout{(#1)}% latexmk will find this if $recorder=0 (however, in that case, it will ignore #1 if it is a .aux or .pdf file etc and it exists! if it doesn't exist, it will appear in the list of dependents regardless)
  \@addtofilelist{#1}% if you want it to appear in \listfiles, not really necessary and latexmk doesn't use this
  \IfFileExists{#1}{}{\typeout{No file #1.}}% latexmk will find this message if #1 doesn't exist (yet)
}
\makeatother

\usepackage[symbol]{footmisc}
\usepackage{epstopdf,gensymb}
\usepackage{graphics}
\usepackage{graphicx,color}

\usepackage{amsfonts, animate,subfigure}
\usepackage{amssymb}
\usepackage{hyperref}
\usepackage{setspace}
\usepackage{indentfirst}
\usepackage{float}
\usepackage{amsfonts}
\usepackage{bm,pbox,bbm}
\usepackage{multirow} 
\usepackage{color}
\usepackage{framed}
\usepackage{amssymb,amsfonts,amsbsy,latexsym,amsmath}
\usepackage{enumerate}
\usepackage{verbatim}

\newtheorem{theorem}{Theorem}
\newtheorem{lemma}[theorem]{Lemma}
\newtheorem{proposition}[theorem]{Proposition}

\newtheorem{definition}[theorem]{Definition}

\newtheorem{assumption}{Assumption}
\renewcommand{\leq}{\leqslant} 
\renewcommand{\geq}{\geqslant}

\newcommand{\eps}{\varepsilon}

\newcommand{\set}[1]{\left\{#1\right\}}

\def\qed{ \hfill $\blacksquare$}  
\newcommand{\cB}{\mathcal{B}}\newcommand{\cC}{\mathcal{C}}
\newcommand{\cE}{\mathcal{E}}
\newcommand{\cG}{\mathcal{G}}

\newcommand{\vD}{\mathbf{D}}

\newcommand{\vS}{\mathbf{S}}

\newcommand{\vc}{\mathbf{c}}
\newcommand{\vd}{\mathbf{d}}

\newcommand{\vs}{\mathbf{s}}
\newcommand{\vx}{\mathbf{x}}
\newcommand{\vy}{\mathbf{y}} 
\newcommand{\expv}{\mathbb{E}}

\DeclareMathOperator{\var}{Var}

\newcommand{\reals}{\mathbb{R}}

\newcommand{\IND}{{\mathbbm{1}}}

\newcommand{\rfun}[1]{r_{uv}(#1)}
\newcommand{\ffun}[2]{f_{uv}(#1,#2)}

\newcommand{\mfun}[3]{\mu(#1 , #2 | #3)}
\newcommand{\sfun}[3]{\sigma(#1,#2|#3)}
\newcommand{\tsfun}[4]{\tilde\sigma(#1,#2|#3,#4)}

\newcommand{\devnt}[2]{\tilde a_n(#1,#2)}

\newcommand{\ruvdag}{\tilde r_{uv}}

\newcommand{\vpar}{\kappa}

\newcommand{\taun}{b\log n }

\newcommand{\Net}{\mathcal{G}}
\newcommand{\esum}{ \sum_E}

\newcommand{\mP}{\mathbf P}
\newcommand{\mM}{\mathbf M}

\newcommand{\mH}{\mathbf H}
\newcommand{\mPi}{\boldsymbol{\Pi}}
\newcommand{\vpi}{\boldsymbol{\pi}}

\newcommand{\Ualph}{U_\alpha(\;\cdot\;, \Net)}

\newcommand{\bra}{\langle}
\newcommand{\ket}{\rangle}

\newcommand{\prob}{\mathbb{P}}
\newcommand{\degp}{\phi}
\newcommand{\strp}{\psi}

\jmlrheading{1}{2017}{**}{7/5}{*/*}{John Palowitch$^*$, Shankar Bhamidi and Andrew B. Nobel}

\ShortHeadings{Community detection in weighted networks}{Palowitch, Bhamidi and Nobel}
\firstpageno{1}

\begin{document}

\title{Significance-based community detection in weighted networks}

\author{\name John Palowitch \email palojj@email.unc.edu \\
	\name Shankar Bhamidi \email bhamidi@email.unc.edu\\
	\name Andrew B.\ Nobel \email nobel@email.unc.edu\\
	\addr Department of Statistics and Operations Research\\
	University of North Carolina at Chapel Hill\\
	Chapel Hill, NC 27599}

\editor{}

\maketitle

\input{section_files/bw_Sec1_Intro.tex}

% Model

\input{section_files/bw_Sec3_Model.tex}

% Core

\input{section_files/bw_Sec4_01_core_algo.tex}

% Theory

\input{section_files/bw_Sec4_02_theory.tex}

% Method

\input{section_files/bw_Sec4_Method3.tex}

% Simulations

\input{section_files/bw_Sec5_Simulations2.tex}

% Data

\input{section_files/bw_Sec6_Data.tex}

%\input{Data_2.tex}

% Discussion

\input{section_files/bw_Sec7_Discussion.tex}

% Appendix start

\appendix

% Appendix_theory

% \input{section_files/bw_Appendix_theory.tex}

% Appendix_nontheory

\input{supplemental1_section_files/CLT_NA.tex}

\input{supplemental1_section_files/Consistency.tex}

\input{section_files/bw_Appendix_nontheory2.tex}

\input{supplemental1_section_files/sim_framework.tex}

\singlespacing
\bibliography{Refs}

\end{document}

%% file: section_files/bw_Sec1_Intro.tex
\begin{abstract}%   <- trailing '%' for backward compatibility of .sty file
	Community detection is the process of grouping strongly connected nodes in a network. Many community detection methods for \emph{un}-weighted networks have a theoretical basis in a null model. Communities discovered by these methods therefore have interpretations in terms of statistical significance. In this paper, we introduce a null for weighted networks called the continuous configuration model. We use the model both as a tool for community detection and for simulating weighted networks with null nodes. First, we propose a community extraction algorithm for weighted networks which incorporates iterative hypothesis testing under the null. We prove a central limit theorem for edge-weight sums and asymptotic consistency of the algorithm under a weighted stochastic block model. We then incorporate the algorithm in a community detection method called CCME. To benchmark the method, we provide a simulation framework incorporating the null to plant ``background" nodes in weighted networks with communities. We show that the empirical performance of CCME on these simulations is competitive with existing methods, particularly when overlapping communities and background nodes are present. To further validate the method, we present two real-world networks with potential background nodes and analyze them with CCME, yielding results that reveal macro-features of the corresponding systems.
\end{abstract}


%%%%%%%%%%%%%%%%%%%%%%%%%%%%%%%%%%%%%%%%%%%%%%%%%%%%%%%%%%%%%%%%%%%%%%%%%%%%%%

\section{Introduction}\label{Intro}

%For decades, the study of mathematical networks has spurred the analysis of relational data from a wide variety of systems. Network-based data analyses have driven advances in areas as diverse as social networks \citep{palla2007quantifying, ahn2007analysis}, systems biology \citep{barabasi2004network}, life sciences \citep{lusseau2004identifying, guimera2005functional}, marketing \citep{reichardt2005ebay}, computer science \citep{hui2007distributed, andersen2012overlapping}, and the study of the internet \citep{eriksen2003modularity, ebel2002scale}. Thorough reviews of the network analysis literature can be found in \cite{jacobs2014unified} and \cite{newman2003structure}.

For decades, the development of graph theory and network science has produced a wide array of quantitative tools for the study of complex systems. Network-based data analysis methods have driven advances in areas as diverse as social science, systems biology, life sciences, marketing, and computer science \cite[cf.][]{palla2007quantifying, barabasi2004network, lusseau2004identifying, guimera2005functional, reichardt2005ebay, andersen12}. Thorough surveys of the network science and methodology literature have been provided by \cite{newman2003structure} and \cite{jacobs2014unified}, among others.

Community detection is a common exploratory technique for networks in which the goal is to find subsets of nodes that are both strongly intraconnected and weakly interconnected \citep{newman2004detecting}. There are many possible definitions of a community, and a broad selection of community detection methods. Nonetheless, community detection can be an important starting point for further inquiry \citep{danon2005comparing}. For instance, community detection has been used to facilitate recommender systems in online social networks \cite[e.g.][]{sahebi11, xin2014book}, and has been used to ``hone in" on regions of genomes (human and otherwise) for a variety of downstream analyses \cite[e.g.][]{cabreros2015detecting, platig2015bipartite, fan2012secom}. Myriad examples of community detection applications can be found in \cite{porter09} and \cite{fortunato10}, and the references therein.

Many community detection methods are based on a null model, which in this context means a random network model without explicit community structure. For un-weighted networks the most common null is the configuration model \citep{bollobas1980probabilistic, bender1974asymptotic} or a related model like that of \cite{chung2002average, chung2002connected}. Historically, the most common approach involving a null model is the use of a node partition score that is large when nodes within the cells of the partition are highly interconnected, relative to what is expected under the null \citep{fortunato10, newman2006modularity}. Arguably the most famous example of such a criterion is modularity, introduced by \cite{newman04}. Various algorithms have been created to search directly for partitions of a network with large modularity \cite[see][]{clauset04, blondel08}, while other approaches use modularity as an auxiliary criterion \cite[see][]{langone2011modularity}. More recent approaches incorporate community-specific criteria which are large when the community exhibits high connectivity, allowing for community \emph{extraction} algorithms \cite[e.g.][]{zhao2011community, lanc11, wilson14}.

Generally speaking, communities found by null-based community detection methods can be said to have exhibited behavior strongly departing from the null. The results of these methods therefore carry a statistical \emph{testing} interpretation unavailable to alternate approaches to community detection, like spectral clustering \citep{white2005spectral, zhang2007identification} or likelihood-based approaches \citep{nowicki01, karrer2011stochastic}. In particular, recent methods put forth by \cite{lanc11} and \cite{wilson14} for binary networks exploit the theoretical properties of the configuration model to detect ``background" nodes that are not significantly connected to any community. These methods incorporate tail behavior of various graph statistics under the configuration model in a way that modularity-based methods do not.

A significant drawback of null-based community detection methodology is that no explicit null model exists for edge-weighted networks. Edge weights are commonplace in network data, and can provide information that improves community detection power and specificity \citep{newman04_2}. While many existing community detection methods have been established for weighted and un-weighted networks alike, due to the absence of an appropriate weighted-network null model, these methods do not provide rigorous significance assessments of weighted-network communities. For instance, the aforementioned method from \cite{lanc11}, called OSLOM, can incorporate edge weights, but uses an exponential function to calculate nominal tail probabilities for edge weight sums, a testing approach which is not based on an explicit null. As a consequence, communities in \emph{weighted} networks identified by OSLOM may in some cases be spurious or unreliable, especially when no ``true" communities exist.

%Nonetheless, few community detection methods (null-based or otherwise) currently handle edge weights. Thus, a common but naive approach is to discretize the graph with a pre-specified weight threshold, and then apply a standard procedure. However, thresholding has been shown ineffective for, or detrimental to, the discovery of communities \citep{aicher2014learning, thomas2011valued}.

%The few methods that directly incorporate edge weights usually involve an extension of an existing binary-network criteria. For instance, a commonly-used weighted version of modularity replaces edge counts with weight sums \citep{newman04_2}. As another example, the OSLOM method \citep{lanc11} uses an exponential function to calculate nominal tail probabilities for edge weight sums, in conjunction with analogous edge count tail probabilities. However, these extensions are fundamentally ad-hoc, since the new criteria with edge weights do not correspond to an explicit null. As a consequence, communities identified by these extended methods may in some cases be spurious or unreliable, especially when no ``true" communities exist.

The key methodological contributions in this article are as follows: (i) we provide an explicit null model for networks with weighted edges, (ii) we present a community extraction method based on hypothesis tests under the null, and (iii) we analyze the consistency properties of the method's core algorithm with respect to a weighted stochastic block model. These contributions provide the beginnings of a rigorous statistical framework with which to study communities in weighted networks. Through extensive simulations, we show that the accuracy of our proposed extraction method is highly competitive with other community detection approaches on weighted networks with both disjoint and overlapping communities, and on weighted networks with background nodes. Importantly, the weighted stochastic block model employed (in both the theoretical and empirical studies) allows for arbitrary expected degree and weighted-degree distributions, reflecting degree heterogeneity observed in real-world networks. To further validate the method, we apply it to two real data sets with (arguably) potential overlapping communities and background nodes. We show that the proposed method recovers sensible features of the real data, in contrast to other methods.

\subsection{Paper organization}

The rest of the paper is organized as follows. We start by introducing general notation in Section \ref{Notation}. In Section \ref{Model} we motivate and state the continuous configuration model. In Section \ref{core-algo}, we introduce a core algorithm to search for communities using multiple hypothesis testing under the model. In Section \ref{theory}, we prove both a central limit theorem and a consistency result for the primary test statistic in the core algorithm. We describe the implementation and application of the core algorithm in Section \ref{Method}, and evaluate its empirical efficacy on simulations and real data in Section \ref{Simulations} and \ref{Data} (respectively). We close with a discussion in Section \ref{Discussion}. 
%{\color{red} possibly extend this section to include a graphic or clear list displaying the sections and their relationships.}

\subsection{Notation and terminology}\label{Notation}

We denote an undirected weighted network on $n$ nodes by a triple $\mathcal{G}:= (N,A,W)$, where $N := \{1,\ldots,n\}$ is the node set with $u,v$ as general elements, $A$ is the adjacency matrix with $A_{uv} = 1$ if and only if there is an edge between $u$ and $v$, and $W$ is the weight matrix with non-negative entries $W_{uv}$ containing edge weights between nodes $u$ and $v$. Note that $A_{uv} = 0$ implies $W_{uv}=0$, but $W_{uv}$ may be zero even when $A_{uv} = 1$. This allows for networks with potentially zero edge weights; for instance, an online social network from which friendship links are edges and message counts are edge weights. The degree of a node $u$ is defined by
$d(u) := \sum_{v \in N} A_{uv}$, and we denote the vector of node degrees by ${\bf d} = (d_1,\ldots,d_n)$.
In an analogous fashion, we define the {\it strength} of a node by
$s(u) := \sum_{v \in N} W_{uv}$, and the strength vector of the network by ${\bf s} = (s(1),\ldots,s(n))$.
The total degree and strength of $\mathcal{G}$ are given by $d_T := \sum_{v \in  N} d(v)$ and 
$s_T := \sum_{v \in  N} s(v)$, respectively.
%U.S. airport networks \cite{rita} and a weighted network of email addresses newly constructed from the publicly available ENRON email corpus \cite{klimt04}.

%% file: section_files/bw_Sec3_Model.tex
\section{The continuous configuration model}\label{Model}

To motivate the null model, we first explain the intuition behind the binary configuration model for unweighted networks.  The binary configuration model for an $n$-node network is based on a given degree vector ${\bf d}$ corresponding to the nodes. Studied originally in \cite{bollobas1980probabilistic} and \cite{bender1974asymptotic}, the model is equivalent to a process in which each node $u$ receives $d(u)$ half-edges, which are paired uniformly-at-random without replacement until no half-edges remain \citep{molloy1995critical}. In other words, the model guarantees a graph with degrees $\vd$ but otherwise uniformly distributed edges. Therefore, given an observed network with degrees $\vd$, a typical draw from the configuration model under $\vd$ represents that network without any community structure. As a result, many community detection methods proceed by identifying node sets having intra-connectivity significantly beyond what is expected under the model. For instance, the modularity measure, introduced by \cite{newman04}, scores node partitions of binary networks according to the observed versus configuration model-expected edge densities of the communities. The methods OSLOM \citep{lanc11} and ESSC \citep{wilson14} use the configuration model to assess the statistical significance of the deviations graph statistics from their configuration model-expected values.
% If networks with self-loops are excluded, the model results in a uniform probability distribution on unweighted graphs with the given degree sequence ${\bf d}$. 

The degrees $\vd$ of the configuration model can be thought of as the nodes' relative propensities to form ties. Chung and Lu made this notion explicit by defining a Bernoulli-based model for a $n$-node unweighted network with a given expected degree sequence \citep{chung2002connected}. Under this model, the probability of nodes $u$ and $v$ sharing an edge is exactly $d(u)d(v)/d_T$. As null models for community detection, the Chung-Lu and configuration are often interchangeable \citep{durak2013scalable}. Indeed, for sparse graphs it can be shown that the probability of an edge between $u$ and $v$ under the configuration model is approximately the Chung-Lu probability.
The \emph{continuous} configuration model, introduced below, extends the spirit of the configuration and Chung-Lu models by taking both observed degrees $\bf d$ and strengths $\bf s$ as node propensities for (respectively) edge connection and edge weight. 

We use the following notation to concisely express the model. Given a vector $\vx$ of dimension $n$, we define for any indices $u,v\in N$ the ratio
\begin{equation}\label{ratio}
r_{uv}(\vx) := \frac{x(u)x(v)}{\sum_{w\in N}x(w)}
\end{equation}
Define $\ruvdag(\vx) := \min\{1,r_{uv}(\vx)\}$. Note that when $\vx$ is a degree sequence $\vd$, $r_{uv}(\vd)$ is the Chung-Lu probability of an edge between nodes $u$ and $v$. Finally, for a vector $\vy$ of dimension $n$, define $f_{uv}(\vx,\vy) := r_{uv}(\vy)/\ruvdag(\vx)$.

\subsection{Model statement}\label{Model:statement}
The continuous configuration model on $n$ nodes has the parameter triple $\theta := ({\bf d}, {\bf s}, \vpar)$, where ${\bf d} \in \{1, 2, 3, \ldots\}^n$ is a degree vector, ${\bf s} \in [0,\infty)^n$ is a strength vector, and $\vpar > 0$ is a variance parameter. Let $F$ be a distribution on the non-negative real line with mean one and variance $\vpar$. The model specifies a random weighted graph 
${\Net} := (N, A, W)$ on $n$ nodes as follows:
\begin{enumerate}
	\item $\prob(A_{uv} = 1) = \ruvdag(\vd)$ independently for all node pairs $u, v \in N$
	\item{For each node pair $u,v$ with $A_{uv} = 1$, generate an independent random variable $\xi_{uv}$ according to $F$, and assign edge weights by: 
		
		$W_{uv} = 
		\begin{cases} 
		f_{uv}(\vd,\vs)\xi_{uv},  & A_{uv} = 1\\
		0,  & A_{uv} = 0\end{cases}$\\
	}
\end{enumerate}
The edge generation defined by step 1 is equivalent to the Chung-Lu model: edge indicators are Bernoulli, with probabilities adjusted by the propensities $\vd$. The weight generation in step 2 mirrors this process. Edge weights follow the distribution $F$, with means adjusted by the propensities $\vs$, through $f(\vd,\vs)$. If $r_{uv}(
\vd)\leq1$ for all $u,v\in N$ (that is, all probabilities are proper), it is easily derived from the model that
\begin{equation}\label{Model:expv2}
P(A_{uv} = 1) = \frac{d(u)d(v)}{d_T}
\ \ \text{ and } \ \ 
\expv\left({W}_{uv}\right) = \frac{s(u)s(v)}{s_T},
\end{equation}
equations which extend the binary-network notion of null behavior to edge weights. The equations in \eqref{Model:expv2} imply that 
\begin{equation}\label{Model:expv1}
\expv\left({D}(u)\right) = d(u) \ \ \text{ and } \ \  \expv\left({S}(u)\right) = s(u)\; \text{ for all }\; u\in N.
\end{equation}
where $D(u)$ and $S(u)$ are the (random) degree and strength of $u$ under the model. Thus, the continuous configuration model can be thought of as null weighted network with given expected degrees and given expected strengths.

\subsection{Use of the null model}\label{Model:specification}
When the \emph{binary} configuration model is used for community detection, the degree parameter of the model is set to the observed degree distribution of the network. In a sense, this is an \emph{estimation} of the nodes' connection propensities under the null. Similarly, to use the continuous configuration model in practice, we derive the parameter $\theta$ from the data at hand. Given an observed network $\Net$, we straightforwardly use the observed degrees and strengths $\bf d$ and $\bf s$ as the first two parameters of the model. The third parameter of the continuous configuration model, $\vpar$, is also computed from the $\Net$, and meant to capture its observed average edge-weight variance. We use the following method-of-moments estimator to specify $\vpar$:
\begin{equation}
\label{Model:var2}
\hat\vpar(\vd,\vs) := \underset{u,v:A_{uv} = 1}{\sum} \left(W_{uv} - f_{uv}(\vd,\vs) \right)^2
/ \underset{u,v:A_{uv} = 1}{\sum}f_{uv}(\vd,\vs)^2
\end{equation}
This estimator is derived as follows. Under the continuous configuration model with $\vd$ and $\vs$,
\begin{equation}
\label{Model:var1}
\text{Var} \left( W_{uv} \, \big| \, A_{uv} = 1 \right)
\ = \  
f_{uv}(\vd,\vs)^2 \, \text{Var}\left({\xi}_{uv}\right) 
\ = \ 
f_{uv}(\vd,\vs)^2 \, \vpar.
\end{equation}
Therefore
\begin{eqnarray}
\expv \left\{ \underset{u,v:A_{uv} = 1}{\sum} \left( W_{uv} - f_{uv}(\vd,\vs) \right)^2 \, \Big| \ A \right\} 
& = &
\underset{u,v:A_{uv} = 1}{\sum}\text{Var} \left( {W}_{uv} \; \big| \; A_{uv} = 1 \right) \nonumber\\
&\;=\; &
\vpar \underset{u,v:A_{uv} = 1}{\sum}f_{uv}(\vd,\vs)^2,\nonumber
\label{Model:var3}
\end{eqnarray}
Dividing through by $\sum_{u,v:A_{uv}=1}f_{uv}(\vd,\vs)$ motivates equation \ref{Model:var2}.

Strictly speaking, the distribution $F$ is also a parameter of the model. However, for testing purposes we do not require a null specification of $F$. As we discuss in the next section, p-values from the model will be based on a central limit theorem that requires only a third-moment assumption on $F$. While estimating $F$ could improve the model's efficacy as a null, in general this would require potentially costly computational procedures, and additional theoretical assumptions that might be difficult to support or verify in practice. The specification of $F$ will be most useful for applications of the model that involve simulations or likelihood-based analyses.

%% file: section_files/bw_Sec4_01_core_algo.tex
\section{Test statistic and update algorithm}\label{core-algo}

In this section we introduce a core testing-based community detection algorithm based on the continuous configuration model. The algorithm allows for a community detection approach which employs iterative node-set updating, following some recently-introduced methods \cite[e.g.][]{lanc11, wilson14}. First, we define a set update as a map $\Ualph: \mathbbm{2}^{N} \mapsto \mathbbm{2}^{N}$, indexed by a parameter $\alpha\in(0,1)$. Given a weighted network $\Net$ and candidate set $B\subseteq N$, the update $U_\alpha(B,\Net)$ outputs a new set $B'$ formed by the nodes from $N$ that have statistically significant association to $B$ at level $\alpha$, after a multiple-testing correction. We now describe $U_\alpha$ in detail.

The connectivity of a single node $u\in N$ to a candidate set $B$ is computed via the simple test statistic
\begin{equation}\label{eq:test_stat}
S(u,B,\Net) := \sum_{v \in B} W_{uv},
\end{equation}
which is the sum of all weights on edges incident with $u$ and $B$. When the observed value of $S(u,B,\Net)$ is much larger than its expectation under the 
continuous configuration model, there is evidence to support an association 
between $u$ and $B$ resulting from some form of ``ground-truth" community structure in the network. We assess the strength of evidence, that is, the significance of $S(u,B,\Net)$, with the p-value
\begin{equation}\label{eq:pvalue}
p(u , B, \Net) := \prob\left( S(u,B,\widetilde{\Net}) > S(u,B,\Net) \right),
\end{equation}
where $\widetilde{\Net}$ is random with respect to $\prob$, the distribution of the continuous configuration model with parameters $\bf d, \bf s$, and $\hat{\vpar}(\bf d, \bf s)$ (see Section \ref{Model:specification}). The update $U_\alpha$ is then:

\vskip0.2cm

\noindent\fbox{\parbox{\linewidth}{
		\textbf{Core update $U_\alpha$}
		\begin{enumerate}[1.]
			\item Given: graph $\Net$ with nodes $N$ and input set $B\subseteq N$
			\item Calculate p-values ${\bf p} := \{p(u,B,\Net):u\in N\}$
			\item Obtain threshold $\tau({\bf p})$ from a multiple-testing procedure
			\item Output set $B'=\{u:p(u,B,\Net)\leq\tau({\bf p})\}$
		\end{enumerate}
	}}
	
\vskip0.2cm

Many methods to compute a multiple-testing threshold $\tau({\bf p})$ are available, the most stringent being the well-known Bonferroni correction. The correction we employ is the false discovery rate (FDR) control procedure of 
\cite{benjamini1995controlling}. Given a set of p-values ${\bf p} := \{p_u\}_{u\in N}$ corresponding 
to $n$ hypothesis tests and a target FDR $\alpha\in(0, 1)$, each p-value $p_u \in {\bf p}$ is associated with an 
adjusted p-value $p_u^\ast := n \, p_u \, / \, j(u)$ where $j(u)$ is the rank of $p_u$ in ${\bf p}$, and $\tau({\bf p}):= \max\{p_u:p_u^\ast\leq\alpha\}$. Benjamini and Hochberg show that, if the p-values corresponding to true null hypotheses are independent, the threshold $\tau({\bf p})$ bounds the expected number of false discoveries at $\alpha$.

The update $U_\alpha$ is an exploratory tool for moving an input set $B$ closer to a ``target" community. Consider that, if the initial set $B$ has a majority group of nodes from some strongly-connected community $C$, the statistic $S(u,B,\Net)$ will be large for $u\in C$, and small otherwise. In this case, $U_\alpha$ applied to $B$ will often return many nodes in $C$, and few nodes in $C^c$. Indeed, ideally, we should expect $U_\alpha(C,\mathbf{D})$ to return $C$, given strong enough signal in the data. This reasoning motivates an algorithm that searches for ``stable communities" $C$ satisfying $U_\alpha(C,\mathbf{D}) = C$. By definition, all interior nodes of a stable community $C$ are significantly connected to $C$, and exterior nodes are not. We define a stable community search procedure, which iteratively applies $U_\alpha$ until convergence:

%The usage of $\Ualph$ in practice is to find ``stable" communities $C$ that satisfy $U_\alpha(C, \Net) = C$. By construction of $\Ualph$, all nodes $u$ in a stable community $C$ have significant observed association to $C$, and all nodes $u\notin C$ do not. Stable communities therefore have the desirable ``strongly intra-connected, weakly inter-connected" property (see the outset of Section \ref{Intro}) in a rigorous, statistical sense that controls for multiple tests. We now define a stable community search (SCS) algorithm which incorporates $\Ualph$:

\vskip0.2cm

\noindent\fbox{\parbox{\textwidth}{
		\textbf{Stable community search (SCS) algorithm}
		\begin{enumerate}[1.]
			\item Given weighted graph $\Net$ with nodes $N$ and initial set $B_1\subseteq N$; set $B_0 := \phi$, $t = 1$
			\item If $B_t = B_{t'}$ for some $t'<t$, terminate. 
			\item Set $B_{t+1} \gets U_\alpha(B_t,\Net)$ and $t \gets t+1$. Return to step 2.
		\end{enumerate}
	}}
	
	\vskip0.2cm
	Since the number of possible node subsets $B_t$ is finite, SCS is guaranteed to terminate. There are some technicalities regarding use of this algorithm, like how to obtain $B_1$, and when in rare cases $t' < t-1$. We relegate resolution of these issues to Section \ref{Method}. For now, the update $U_\alpha$ and SCS raise two theoretical questions:
\begin{enumerate}[1.]
	\item Is the p-value $p(u,B,\Net)$ analytically tractable? If not, is there a useful distributional approximation based on the continuous configuration model?
	\item Consistency: with what power can SCS detect ground-truth community structure?
\end{enumerate}
These questions are the focus of the next section.

%% file: section_files/bw_Sec4_02_theory.tex
\section{Theoretical Results}\label{theory}
We now address the theoretical questions raised at the end of the previous section by analyzing the distribution of the test statistic $S(u,B,\Net)$ under the continuous configuration model (for question 1) and an appropriate alternative model with planted community structure (for question 2). Both analyses have an asymptotic setting consisting of a sequence of random weighted networks. Denote this sequence by $\{\Net_n\}_{n>1}$. If $\Net_n$ is a continuous configuration model with parameters $\theta := (\vs, \vd, \vpar)$, the following proposition gives general expressions for the mean and standard deviation of $S(u,B,\Net_n)$:
\begin{proposition}\label{prop:sub}
	\label{sub-mean-sd}
	Let $\Net= (N, A, W)$ be a random network generated by the continuous configuration model with parameters $\theta = (\vs,\vd,\vpar)$. For any $(u,B)\in N\times 2^{ N}$, let $\mfun{u}{B}{\theta}$ and $\sfun{u}{B}{\theta}$ be, respectively, the mean and standard deviation of $S(u,B,\Net)$ under $\Net$. Then
	\begin{equation}\label{eq:mean_fun}
	\mfun{u}{B}{\theta} \equiv \mfun{u}{B}{\vs} = \underset{v\in B}{\sum}r_{uv}(\vs)
	\end{equation}
	and
	\begin{equation}\label{eq:var_fun}
	\sfun{u}{B}{\theta}^2 = \underset{v\in B}{\sum}r_{uv}(\vs)f_{uv}(\vd,\vs)\left(1 - \ruvdag(\vd) + \vpar\right)
	\end{equation}
\end{proposition}
The proof, given in Appendix \ref{Prop1proof}, follows from easy calculations with the model's generating procedure (see Section \ref{Model:statement}). All theoretical results will make use of the expressions defined in equations \ref{eq:mean_fun} and \ref{eq:var_fun}.

\subsection{Asymptotic Normality of $S(u,B,\Net)$}\label{Method:CLT}
A central limit theorem under the null model is now established for $S(u,B,\Net)$, yielding a closed-form approximation for the p-value in equation  \eqref{eq:pvalue}. This result is motivated by the fact that, under most non-trivial null parameter specifications, the distribution of $S(u,B,\Net)$ is not analytically tractable. 

In the setting of the theorem, for any $n>1$, a random network $\Net_n$ is generated by a continuous configuration model with parameter $\theta_n := ({\bf d}_n, {\bf s}_n, \vpar_n)$ and common weight distribution $F$. The following regularity conditions are required on the sequence  $\{\theta_n\}_{n>1}$. Let $\lambda_n$ denote the average entry of ${\bf d}_n$, (which is the average 
expected degree of ${\Net}_n$).  For each $r \geq 0$ let 
$L_{n, r} := n^{-1} \sum_{u \in  N} (d_n(u)/\lambda_n)^r$ be the normalized 
$r^{\text{th}}$-moment of ${\bf d}_n$.  Note that $L_{n, 1} = 1$.
The regularity conditions are then as follows:

\noindent
\begin{assumption}\label{assumption:power}
	Define $e_n(u|\beta):=s_n(u)/d_n(u)^{1+\beta}$. There exists $\beta > 0$ such that
	\[
	0 < \underset{n\rightarrow\infty}{\liminf} \; \min_{u\in N}e_n(u|\beta)
	\;\text{ and }\;
	\limsup_{n\rightarrow\infty}\; \max_{u \in  N} e_n(u|\beta) < \infty.
	\]
\end{assumption}

\noindent
\begin{assumption}\label{assumption:moment}
	Let $\beta$ be as in Assumption \ref{assumption:power}. There exists $\eps > 0$ such that, for both $r = 4\beta + 2$ and $r = 4\beta + 2 + \eps$,
	\[
	0 < \underset{n \rightarrow \infty}{\liminf} \; L_{n, r}
	\;\text{ and }\;
	\limsup_{n\rightarrow\infty} \; L_{n, r} < \infty
	\]
\end{assumption}

\noindent
\begin{assumption}\label{assumption:ruv}
	$\underset{n\rightarrow\infty}{\limsup}\;\underset{u,v\in N}{\sup}r_{uv}(\vd_n)<\infty$.
\end{assumption}

\noindent
\begin{assumption}\label{assumption:distribution}
	The sequence $\{\vpar_n\}_{n \geq1}$ is bounded away from zero and infinity, 
	and $F$ has finite third moment.
\end{assumption}

%These parameter conditions have a trade-off relationship. The value of $\beta$ imposes regularity, via the power-law condition, on the relationship between the degrees and strengths; $\beta =0$ corresponds to the regime where $O(s_n(u)) = O(d_n(u))$ for all $u$. As $\beta$ increases, we require the stability of increasingly higher moments for the normalized degree sequences, via the moment condition. 

Assumption \ref{assumption:power} reflects the common relationship between strengths and degrees 
in real-world weighted networks \citep{barrat04, clauset09}.  Assumptions \ref{assumption:moment}-\ref{assumption:ruv} are needed to control
the extremal behavior of the degree distribution. They exclude, for instance, cases with 
a few nodes having $d_n(u)\asymp n$ and the remaining nodes having $d_n(u) = O(1)$. We note that the Assumption \ref{assumption:moment} becomes more stringent as $\beta$ increases, since as $\beta$ increases the strength-degree power law becomes more severe.
% CLT Theorem
\begin{theorem} \label{thm:CLT}
	For each $n>1$, let $\Net_n$ be generated by the continuous configuration model with parameter $\theta_n$ and weight distribution $F$. Suppose $\{\theta_n\}_{n\geq1}$ and $F$ satisfy Assumptions \ref{assumption:power}-\ref{assumption:distribution}.  Fix a node sequence $\{u_n\}_{n\geq1}$ with $u_n\in N$ and a positive integer sequence $\{b_n\}_{n\geq1}$ with $b_n\leq n$. Suppose $d_n(u_n)b_n/n\rightarrow\infty\text{ as }n\rightarrow\infty$. Let $B_n\subseteq N$ be a node set chosen independently of ${\Net}_n$ according to the uniform distribution on all sets of size $b_n$. Then
	\begin{equation}\label{eq:CLT}
	\frac{S(u_n,B_n,\Net_n) - \mu_n(u_n,B_n|\theta_n)}{\sigma_n(u_n,B_n|\theta_n)}\Rightarrow\mathcal{N}(0,1)\;\text{ as }\;n\rightarrow\infty
	\end{equation}
\end{theorem}
% End CLT Theorem

The proof is given in Appendix \ref{clt-proof}. Essentially, Theorem \ref{thm:CLT} says that $S(u,B,\Net)$ is asymptotically Normal provided that $B$ is ``typical" and that $d(u)$ and $B$ are sufficiently large. The theorem justifies the following approximation of the p-value in \eqref{eq:pvalue}:
\begin{equation}\label{eq:pvalue_approx}
p(u,B) \approx 1 - \Phi\left(\dfrac{S(u,B,\Net) - \mfun{u}{B}{\theta}}{\sfun{u}{B}{\theta}}\right)
\end{equation}
Above, $\theta = (\vd,\vs,\hat\kappa(\vd,\vs))$ is specified from $\Net$, as described in Section \ref{Model:specification}.

\subsection{Consistency of SCS}\label{Consistency}

In this section, we evaluate the ability of the SCS algorithm to identify true communities in a planted-community model. Explicitly, we consider a sequence of networks $\{\Net_n\}_{n>1}$ where each network in the sequence is generated by a weighted stochastic block model (WSBM). The WSBM we employ is similar to that presented in \cite{aicher2014learning}, but is generalized to include node-specific weight parameters. In other words, it is ``strength-corrected" as well as degree-corrected, in a manner analogous to the original degree-corrected SBM \citep{coja2009finding}. The proofs of Theorem \ref{initial-to-community} and Theorem \ref{thm:consistency} are given in Appendix \ref{consistency-proofs}.

\subsubsection{The weighted stochastic block model}\label{WSBM}
For fixed $K > 1$, we define a $K$-block WSBM on $n > 1$ nodes as follows. Let $\vc_n$ be a community partition vector with $c_n(u)\in\{1,\ldots,K\}$ giving the community index of $u$. Denote community $i$ by $C_{i,n}:=\{u:c_n(u) = i\}$. Define $\pi_{i,n} := n^{-1}|C_{i,n}|$ with $\vpi_n$ the associated vector. Let $\mP$ and $\mM$ be fixed $K\times K$ matrices with non-negative entries encoding intra- and inter-community baseline edge probabilities and edge weight expectations, respectively. Let $\phi_n$ and $\psi_n$ be arbitrary $n$-vectors with positive entries, which are parameters giving nodes individual propensities to form edges and assign weight (separately from $\mP$ and $\mM$). To ensure proper edge probabilities, we assume that $\max(\phi_n)^2\max(\mP)\leq 1$. For identifiability, we assume the vectors $\phi_n$ and $\psi_n$ sum to $n$. Finally, let $F$ be a distribution on the positive real line with mean 1 and variance $\sigma^2\geq 0$. The WSBM can then be specified as follows:
\begin{enumerate}
	\item Place an edge between nodes $u$ and $v$ with probability $\mathbb{P}_n(A_{uv} = 1) = r_{uv}(\phi_n)\mP_{c_n(u)c_n(v)}$, independently across node pairs.
	\item{For node pair $u,v$ with $A_{uv} = 1$, generate an independent random variable $\xi_{uv}$ according to $F$. Determine edge weight $W_{uv}$ by:
		\[
		W_{uv} = \begin{cases}
		\ffun{\psi_n}{\phi_n}\mM_{c_n(u)c_n(v)}\xi_{uv},& A_{uv} = 1\\
		0,&A_{uv} = 0\\
		\end{cases}
		\]
	}
\end{enumerate}

The many parameters involved with this model allow for node heterogeneity and community structure. When $\mP$ and $\mM$ are proportional to a $K\times K$ matrix of ones, the WSBM reduces to the continuous configuration model with parameters $\vd \propto \phi$, $\vs \propto \psi$, and $\vpar = \sigma^2$. Community structure is introduced in the network by allowing the diagonal entries of $\mP$ and $\mM$ to be arbitrarily larger than the off-diagonals.

\subsubsection{Consistency theorem}\label{Consistency-theorem}
The consistency analysis of SCS involves a sequence of random networks $\{\mathcal G_n\}_{n>1}$, where $\mathcal G_n$ is generated by a $K$-community WSBM. In this setting, we incorporate an additional parameter $\rho_n$, and let $\mP_n:=\rho_n\mP$ replace $\mP$ for each $n>1$. This lets us distinguish the role of the asymptotic order of the average expected degree, defined $\lambda_n:=n\rho_n$, from the profile of edge densities within and between communities ($\mP$). Importantly, our results require only that $\lambda_n/\log n\rightarrow\infty$, reflecting the sparsity of real-world networks. Throughout this section, we denote the vector of (random) strengths from $\Net_n$ by $\vS_n$.

We now define an explicit notion of consistency in terms of the SCS algorithm. Recall from Section \ref{core-algo} that for fixed FDR $\alpha\in(0,1)$, a stable community in a network $\Net_n$ is defined as a node set $C\subseteq N$ satisfying $U_\alpha(C,\Net_n) = C$.
\begin{definition}
	\label{consistency}
	We say that SCS is consistent for a sequence of WSBM random networks $\{\Net_n\}_{n>1}$ if for any FDR level $\alpha\in(0,1)$, the probability that the true communities $C_{1,n},\ldots,C_{K,n}$ are stable approaches 1 as $n\rightarrow\infty$.
\end{definition}

To assess the conditions that allow a target set $C$ to be a stable community, we seek more general conditions under which the update $\Ualph$ outputs $C$ given any initial set $B$. If $U_\alpha(B,\Net_n) = C$, all nodes $u\in C$ must have significant connectivity to $B$, as judged by the p-value approximation defined in \ref{eq:pvalue_approx}. It is clear from that p-value expression that, for the update to return $C$, the test statistic $S(u,B,\Net_n)$ must be significantly larger than $\mu(u,B|\vS_n)$, its expected value under the continuous configuration model. Therefore, our first result hinges on asymptotic analysis of that deviation, which we denote by
\begin{equation}\label{Dev-def}
A(u,B,\Net_n) := S(u,B, \Net_n) - \mu_n(u,B|\vS_n).
\end{equation}
The asymptotics of $A(u,B,\Net_n)$ depend on its \emph{population} version, in which all random quantities are replaced with their expected values under the WSBM. Let $\vs_n$ be the expected value of $\vS_n$ under $\Net_n$. We define the (normalized) population version of $A(u,B,\Net_n)$ by
\begin{equation}\label{dev-def}
\devnt{u}{B}:=\lambda_n^{-1}\left(\expv S(u,B,\Net_n) - \mu_n(u,B|\bar \vs_n)\right),
\end{equation}
where $\lambda_n$ is the order of the average expected degree. The value $\devnt{u}{B}$ is crucial to the primary condition of Theorem \ref{initial-to-community}. Given a sequence of initial sets $\{B_n\}_{n>1}$ and target sets $\{C_n\}_{n>1}$, Theorem \ref{initial-to-community} establishes that $U_\alpha(B_n,\Net_n) = C_n$ with probability approaching 1 if $\devnt{u}{B}$ is bounded away from zero, and is positive if and only if $u\in C_n$. The theorem requires the following two assumptions:

%%%%%%%%%
%%% Assumptions
%%%%%%%%%

\begin{assumption}\label{bounded-parameter-assumption} There exist constants $m_+>m_->0$ such that, for all $n>1$, the entries of $\phi_n$, $\psi_n$, $\mP$, $\mM$, and $\vpi_n$ are all bounded in the interval $[m_-,m_+]$.
\end{assumption}
\begin{assumption}\label{bounded-weight-assumption} $F$ is independent of $n$ and has support $(0,\eta)$ with $\eta<\infty$.
\end{assumption}
Assumption \ref{bounded-parameter-assumption} is standard in consistency analyses involving block models \cite[e.g.][]{zhao2012consistency, bickel2009nonparametric}. Assumption \ref{bounded-weight-assumption} allows the use of Bernstein's inequality throughout the proof, but may be relaxed if there are constraints on the moments of $F$ allowing the use of a similar inequality. We now state Theorem \ref{initial-to-community}, the proof of which is given in Appendix \ref{consistency-proofs}.

%%%%%%%%%
%%% Lemma
%%%%%%%%%

\begin{theorem}\label{initial-to-community}
	Fix $K>1$. For each $n > 1$, let $\Net_n$ be a $n$-node random network generated by a $K$-community WSBM with parameters satisfying Assumptions \ref{bounded-parameter-assumption} - \ref{bounded-weight-assumption}. Suppose $\lambda_n/\log n\rightarrow\infty$. Let $\{B_n\}_{n>1}$, $\{C_n\}_{n>1}$ be sequences of node sets satisfying the following: there exist constants $q\in(0,1]$ and $\Delta>0$ such that for all $n$ sufficiently large, $|B_n|,|C_n|\geq qn$, and
	\begin{equation}\label{favors}
	\tilde a_n(u,B_n)\geq \Delta,\;\;u\in C_{n},\;\;\text { and }\;\;
	\tilde a_n(u,B_n)\leq -\Delta,\;\;u\notin C_{n}.
	\end{equation}
	Then if the update $U_\alpha$ uses the p-value approximation given in Equation \eqref{eq:pvalue_approx},
	\[
	\mathbb{P}_n\big(U_\alpha(B_n, \cG_n) = C_{n}\big)\rightarrow1\text{ as }n\rightarrow\infty.
	\]
\end{theorem}
To prove the consistency of SCS, we show that condition \ref{favors}, when it involves the community sequence, is guaranteed by a concise condition on the model parameters. Let $\tilde\pi_{i,n} := \sum_{v\in C_{i,n}}\psi_n(v)$, and let $\tilde \vpi_n$ be the vector of $\tilde \pi_{i,n}$'s. The consistency theorem requires the following additional assumption, an analog to which can be found in \cite{zhao2012consistency} for consistency of modularity under the degree-corrected SBM:
	\begin{assumption}\label{constant-psi} $\tilde \vpi_n \equiv \tilde \vpi$ does not depend on $n$.
	\end{assumption}
Assumption \ref{constant-psi} is made mainly for clarity. Without it, the condition in \eqref{M-favors} of Theorem \ref{thm:consistency} (below) must hold for sufficiently large $n$, something which is inconsequential to the proof. Define $\mH := \mP\cdot\mM$, the entry-wise product. Note that when $\phi$ and $\psi$ are proportional to 1-vectors, $\expv(W_{uv}) = \mH_{c(u)c(v)}$ for all $u,v\in N$. Thus, the interpretation of $\mH$ is as the baseline inter/intra-community weight expectations after integrating out edge presence. Defining $\tilde{\mPi} := \tilde \vpi\tilde \vpi^t$, we state the consistency theorem:
\begin{theorem}
	\label{thm:consistency}
	Fix $K>0$. Let $\{\Net_n\}_{n>1}$ be a sequence of networks generated by a $K$-community WSBM satisfying Assumptions \ref{bounded-parameter-assumption}-\ref{constant-psi}. Suppose that the matrix
	\begin{equation}\label{M-favors}
	\mathcal{M} := \mH - \frac{\mH\tilde{\mPi}\mH}{\tilde \vpi^t\mH\tilde \vpi}
	\end{equation}
	has positive diagonal entries and negative off-diagonal entries. If $\lambda_n/\log n\rightarrow\infty$, SCS is consistent for $\{\Net_n\}_{n>1}$.
\end{theorem}
%The import of condition \eqref{M-favors} is most clear when $K = 2$, when it reduces to the requirement that the on-diagonals of $\mH$ exceed the (single) off-diagonal. Notice that this signal need not be present in both $\mP$ and $\mM$. For instance, the condition would be satisfied if $\mH$ is a scalar multiple of $\mM$, that is, if $\mP$ is proportional to the $\mathbf{1}$-matrix. This entails that SCS is consistent even when the edge structure of $\Net_n$ is Erd\H{o}s-Renyi, as long as the edge weight signal is assortative. Of course, the opposite also holds, namely that SCS is consistent even when assortative community signal is only present in $\mP$.

The proof of Theorem \ref{thm:consistency} is given in Appendix \ref{consistency-proofs}. Understanding of condition \ref{M-favors} begins with the consideration of the case $K = 2$, when it reduces to the requirement that $\mH_{11}\mH_{22}>\mH_{12}^2$. More generally, and broadly speaking, the matrix $\mathcal{M}$ reveals whether or not appropriate signal exists in the model, with respect to the continuous configuration null. Notice that this signal need not be present in both $\mP$ and $\mM$. For instance, the condition can be satisfied even if $\mH$ is a scalar multiple of $\mM$, that is, if $\mP$ is proportional to the $\mathbf{1}$-matrix. This entails that SCS is consistent even when the edge structure of $\Net_n$ is Erd\H{o}s-Renyi, as long as the edge weight signal (encoded in $\mM$) is properly assortative. Of course, the opposite also holds, namely that SCS is consistent even when assortative community signal is only present in $\mP$.

\subsubsection{Connection to weighted modularity and related work}\label{theory-wm}
The conditions of Theorem \ref{initial-to-community} and Theorem \ref{thm:consistency} have a deep relationship to the modularity measure, discussed in Section \ref{Model}.  Explicitly, let the \emph{weighted} modularity (WM) be the modularity metric with degrees replaced by strengths, as introduced in \citep{newman04_2}. For fixed $n>1$, let $\vc$ be any partition of $ N$. Define $K := \max\{\vc\}$ and $C_u:=\{v:c(v) = c(u)\}$. Then the (random) WM of $\vc$ on $\Net_n$ can be written
\begin{align*}
Q^w(\vc, \Net_n) & := \frac{1}{S_{n, T}}\sum_{uv\in N}\left\{W_{uv} - r_{uv}(\vS_n)\right\}\IND\{c(u) = c(v)\}\\
& = \frac{1}{S_{n, T}}\sum_{i=1}^K\sum_{c(u) = c(v)}W_{uv} - r_{uv}(\vS_n) = \frac{1}{S_{n, T}}\sum_{u\in N}\sum_{v\in C_u}W_{uv} - r_{uv}(\vS_n)\\
& = \frac{1}{S_{n, T}}\sum_{u\in N}S(u,C_u,\Net_n) - \mu_n(u,C_u|\vS_n) = \frac{1}{S_{n, T}}\sum_{u\in N}A(u, C_u,\Net_n)
\end{align*}
Thus, the contribution of $u$ to WM with its assignment $C_u$ is precisely the random association from $u$ to $C_u$. Writing the population WM as $\bar q_n^w(\cC) := n^{-1}\sum_u\tilde a_n(u,C_u)$, it is easily shown that condition \eqref{M-favors} implies $q_n^w$ is maximized by $\cC_n$, the true community partition.

The consistency analysis of the (binary) modularity metric under the degree-corrected SBM, provided by \cite{zhao2012consistency}, similarly hinges on maximization of population modularity. It is unsurprising, then, that the parameter condition for their result can be (analogously) expressed as a fixed $K\times K$ matrix having positive diagonals and negative off-diagonals. In fact, if the WSBM parameter $\mM$ is proportional to a matrix of 1s, and the parameter $\psi$ is a scalar multiple of $\phi$, condition \ref{M-favors} in Theorem \ref{thm:consistency} is equivalent to the parameter assumptions on modularity consistency in \cite{zhao2012consistency}. Furthermore, their analysis also requires that $\lambda_n/\log n\rightarrow\infty$. However, both the definition of consistency and proof approach for the theorems in this section are entirely novel.

\begin{comment}

\begin{proposition}
	\label{favors-prop}
	Fix $K>1$. For each $n>1$, let $\Net_n$ be a random network generated by a $K$-WSBM with $n>1$ nodes, and parameters satisfying Assumptions \ref{community-size-assumption} - \ref{bounded-weight-assumption}. For fixed $i\in[K]$, suppose condition \eqref{favors} holds for the sequence $\{C_{i,n}\}_{n>1}$. Define $\{B_n\}_{n>1}$ by
	\[
	B_n = B'_n\cup C_n
	\]
	with $B'_n\subseteq \cC_{i,n}^c$ and $C_n\subseteq\cC_{i,n}$. If $n^{-1}|B'_n|\rightarrow0$, then there exists $N$ such that for all $n>N$, condition \eqref{favors} holds for the sequence $\{B_n\}_{n>1}$.
\end{proposition}

\end{comment}

%% file: section_files/bw_Sec4_Method3.tex
\section{The Continuous Configuration Model Extraction method}\label{Method}
In the previous section, we established an asymptotic result showing that ground-truth communities are, with high probability, fixed points of the SCS algorithm. This result demonstrates the in-principle sensibility of the algorithm. In practice, we must rely on local, heuristic algorithms for initialization and termination, as with other exploratory methods. For instance, $k$-means is often used to initialize the EM algorithm, and modularity can be locally maximized through agglomerative pairing \citep{clauset04}. We incorporate SCS in a general community detection method for weighted networks entitled Continuous Configuration Model Extraction (CCME), written in loose detail as follows:
\vskip0.2cm
\noindent\fbox{\parbox{\textwidth}{
\noindent \textbf{The CCME Community Detection Method for Weighted Networks}
\begin{enumerate}[1.]
\item Given an observed weighted network $\Net$, obtain initial node sets $\mathcal{B}_1\subseteq 2^{N}$.
\item Apply SCS to each node set in $\mathcal{B}_0$, resulting in fixed points $\mathcal{C}$.
\item Remove sets from $\mathcal{C}$ that are empty or redundant.
\end{enumerate}
}}
\vskip0.2cm
\noindent These steps are described in more detail below. Importantly, the method has no connection to any graph-partition criteria. It proceeds solely by the SCS algorithm, which assesses communities independently. This allows CCME to adaptively return communities that share nodes (``overlap"), and, through the multiple testing procedure, ignore nodes not significantly connected to any stable communities (``background").
\subsection{Step 1: Initialization}
Just as principled mixture-models can be initialized with heuristic methods like $k$-means, it is possible to initialize CCME with partition-based community detection method. However, we have observed this approach to perform somewhat poorly in practice. Instead, we initialize with a novel search procedure based on the continuous configuration model. For fixed nodes $u,v \in N$, we define
\[
z_u(v) := \max\left\{\dfrac{W_{uv} - f_{uv}(\vs,\vd)}{\sqrt{\theta}f_{uv}(\vs,\vd)},\;0\right\}
\]
The measure $z_u(v)$ acts like a truncated $z$-statistic, quantifying the extremity of the weight $W_{uv}$. The initial node set corresponding to $u$ is formed by sampling $d(u)$ nodes with 
replacement from $N$ with probability proportional to $z_u(v)$. The intuition behind this procedure is that if $u$ is part of a highly-connected node set $C$, then $z_u(v)$ for nodes $v\in C$ will be larger (on average) than for other nodes.

\subsection{Step 2: Application of SCS}\label{ss:FPS}
Recall that, given an initial set $B_1$, SCS proceeds (via the update $U_\alpha$) along a sequence of sets $B_2,B_3,\ldots,B_t,\ldots$ until $B_t = B_{t'}$ for some $t'< t$. Since the number of possible node subsets is finite, SCS is guaranteed to terminate in one of two states:
\begin{enumerate}[1.]
	\item A stable community $C$, satisfying $U_\alpha(C,\Net) = C$.
	\item A stable sequence of communities $C_1, \ldots, C_J$ satisfying 
	\[U_\alpha(C_1,\Net) = U_\alpha(C_2, \Net) = \ldots = U_\alpha(C_J,\Net) = U_\alpha(C_1,\Net).\]
\end{enumerate}
\noindent In practice, on empirical and simulated data, case 1 is the majority. In case 2, SCS does not result in a clear-cut community. However, a stable sequence may still be of practical interest if the constituent sets have high overlap. In Appendix \ref{app:cycle}, we give a routine to re-initialize or terminate SCS when it encounters a stable sequence.

\subsection{Step 3: Filtering of $\mathcal{C}$}\label{Method:extract}
The CCME community detection method returns a final collection of communities $\mathcal{C}$, containing the results of the SCS algorithm for each initial set in $\mathcal{B}_0$. By default, we remove any empty or duplicate sets from $\mathcal{C}$. In some applications, pairs of sets in $\mathcal{C}$ will have high Jaccard similarity. In Appendix \ref{app:filtering}, we detail a method of pruning these near-duplicates from $\mathcal{C}$. Additionally, in Appendix \ref{app:filtering}, we describe routines to suppress the application of SCS to initial sets that are ``weakly" intra-connected, or with high overlap to already-extracted communities. These routines greatly reduce the runtime of CCME, and, on some simulated networks, improve accuracy.

\textbf{Remark:} We note that the parameter $\alpha$, used in the set update operation $U_\alpha$, must be specified by the user of CCME. Having a natural interpretation as the false-discovery rate for each update, $\alpha$ was set to $0.05$ for all simulations and real data analyses introduced in this paper. We found that $\alpha = 0.05$ was a universally effective default setting, and that CCME's results change negligibly for other values of $\alpha$ within a reasonable window.

%% file: section_files/bw_Sec5_Simulations2.tex
\section{Simulations}\label{Simulations}
This section contains a performance analysis of CCME and existing methods on a benchmarking simulation framework. Simulated networks are generated from the Weighted Stochastic Block Model (see Section \ref{WSBM}), with slight modifications to include overlapping communities and background nodes, when necessary. The performance measures, competing methods, simulation settings, and results are described below.

\subsection{Performance measures and competing methods}\label{Simulations:methods}
To assess the performance of a community detection method  the various methods, we use three measures:
\begin{enumerate}
	\item \textbf{Overlapping Normalized Mutual Information (oNMI):} 
	Introduced by \cite{lancichinetti2009detecting}, oNMI is an information-based measure between 0 and 1 
	that approaches 1 as two covers of the same node set become similar and equals 1 when they are the same. 
	From a method's results, we calculate oNMI with respect to the true communities \emph{only} 
	for the nodes the method placed into communities.
	
	\item \textbf{Community nodes in background (\%C.I.B.):} The percentage of true community nodes incorrectly assigned to background.
	
	\item \textbf{Background nodes in communities (\%B.I.C.):} The percentage of true background nodes (if present) incorrectly placed into communities.
\end{enumerate}
In addition to CCME, two other weighted-network methods capable of identifying overlapping nodes are assessed. One of these is OSLOM \citep{lanc11}, described in Section \ref{Intro}. The other is SLPAw, a weighted-network version of an overlapping label propagation algorithm \citep{xie11}. Also included are four commonly used score-based methods implemented in the $\mathtt{R}$ package $\mathtt{igraph}$ \citep{igraph}: Fast-Greedy, which performs approximate modularity optimization via a hierarchical agglomeration \citep{clauset04}; Louvain, an approximate modularity optimizer that proceeds through node membership swaps \citep{blondel08}; Walktrap, an agglomerative algorithm that locally maximizes a score based on random walk theory \citep{pons2006computing}; Infomap, an information-flow mapping algorithm that uses random walk transition probabilities \citep{rosvall2008maps}. 

\textbf{Remark.} Being extraction methods, only CCME and OSLOM naturally specify background nodes, via testing. As such, we will often make direct comparative comments between OSLOM and CCME with respect to background node handling. For other methods, we take as background any nodes in singleton communities. However, these methods almost never returned singleton communities, even when the simulation had weak or non-existent signal.

\subsection{Simulation settings and results}
We now give an overview of the simulation procedure for the benchmarking framework. A complete account is given in Appendix \ref{sim-framework}. We first describe ``default" parameter settings of the WSBM; in the simulation settings below, individual parameters are toggled around their default values, to reveal the dependence of the methods to those parameters. At each unique parameter setting, 20 random networks were simulated. The points in each plot from Figure \ref{fig:sbm} show the average performance measure of the methods over the 20 repetitions. 

The default WSBM setting has the number of nodes at $n = 5,000$. The community memberships were set by obtaining community sizes from a power law, then assigning nodes uniformly at random. This process produced approximately 3 to 7 communities per network. Full details are provided in Appendix \ref{sim-framework}. Recall the parameters $\mP$ and $\mM$, which induce baseline intra- and inter-community edge and weight signal. In the default setting, these matrices have off-diagonals equal to 1 and diagonals equal to constants $s_e = 3$ and $s_w = 3$ (respectively). In some simulation settings, overlapping and background nodes are added (as described later in this section), but the default setting includes neither overlap nor background. 

\textbf{Common parameter settings.} For all simulated networks (regardless of the setting), the node-wise edge parameters $\phi$ were drawn from a power law to induce degree heterogeneity. The parameter $\phi$ is scaled so that the expected average degree of each network was equal to $\sqrt{n}$, which induces sparsity in the network. The parameter $\psi$ is set by the formula $\psi = \phi^{1.5}$ to ensure a non-trivial relationship between expected degrees and expected strengths (see Appendix \ref{sim-framework}).

\subsubsection{Networks with varying signal levels}\label{Simulations:disjoint}
The first simulation setting tested the methods' dependence on $s_e$ and $s_w$. These values were moved along an even grid on the range $[1, 3]$. Plots A-1 and B-1 in Figure \ref{fig:sbm} show the performance measure results when $s_w$ is fixed at 3, plots A-2 and B-2 show results when $s_e = 3$, and plots A-3 and B-3 show results when $s_e$ and $s_w$ are moved along $[1, 3]$ together. Many methods had large oNMI scores in this simulation setting. We transformed the oNMI scores using the function 
\[
\text{t-oNMI}_a(x) := (\tfrac{1}{1-x + a} - \tfrac{1}{1+a})/(\tfrac{1}{a} - \tfrac{1}{1+a})
\]
with $a = 0.05$. This is a monotonic, one-to-one transformation from $[0,1]$ to itself, which stretches the region close to 1, allowing a clearer comparison between the methods' performances. CCME consistently out-performed all competing methods, especially when either the edge or weight signal was completely absent.

The plots in row B show that when either $s_e$ or $s_w$ were near 1, OSLOM and CCME assigned many background nodes. This is consistent with these methods' unique abilities to leave nodes unassigned when they are not significantly connected to communities. That said, \%C.I.B. can be seen as a measure of sensitivity, since ideally no nodes would be assigned to background when any signal is present. In this regard, CCME outperformed OSLOM across the range of model parameters.

\subsubsection{Networks with overlapping communities}\label{Simulations:overlapping}
The second setting involved networks with overlapping nodes. To add overlapping nodes to the default network, two parameters were introduced: $o_n$, the number of overlapping nodes, and $o_m$, the number of memberships for each overlapping node. The particular overlapping nodes and community memberships were chosen uniformly-at-random. This closely follows a simulation approach taken by \cite{lanc11}. Plots C-1 and C-2 show performance results from the setting with $o_n$ moving away from 0 and $o_m = 2$. Plot C-3 shows results from the setting with $o_n = 500$ and $o_m\in\{1,\ldots,4\}$. We find that CCME consistently outperforms all methods in terms of accuracy (oNMI), and outperforms OSLOM in terms of sensitivity (\%C.I.B.).

\subsubsection{Networks with overlapping communities and background nodes}\label{Simulations:background}
The final simulation setting involved networks with both overlap and background nodes. The number of background nodes was fixed at 1,000, and number of community nodes varied from $n = 500$ to $n = 5,000$. For each network, $o_n = n / 4$ nodes were randomly chosen to overlap $o_m = 2$ communities (also chosen at random). Background nodes were created by first simulating the $n$-node community sub-network, and then generating the 1,000-node background sub-network according to the continuous configuration model, using empirical degrees and strengths from the community sub-network. The complete details of this procedure are given in Appendix \ref{sim-framework}.

The results of this simulation setting are shown in row D from Figure \ref{fig:sbm}. From plot D-1, we see that OSLOM and CCME had the highest oNMI scores, 
favoring OSLOM when the number of community nodes decreased. 
Because this simulation setting involved background nodes, the \%B.I.C. metric is relevant, and can be taken as a measure of specificity: ideally, nodes from the background sub-network should be excluded from communities. From plot D-2, we see that methods incapable of assigning background had \%B.I.C. equal to 1. We found that CCME correctly ignored background nodes as the network size increased, 
whereas OSLOM became increasingly \emph{anti}-conservative for larger networks. Furthermore, CCME again had lower \%C.I.B. than OSLOM.

\begin{figure}
	\centering
    \includegraphics[scale = 0.29]{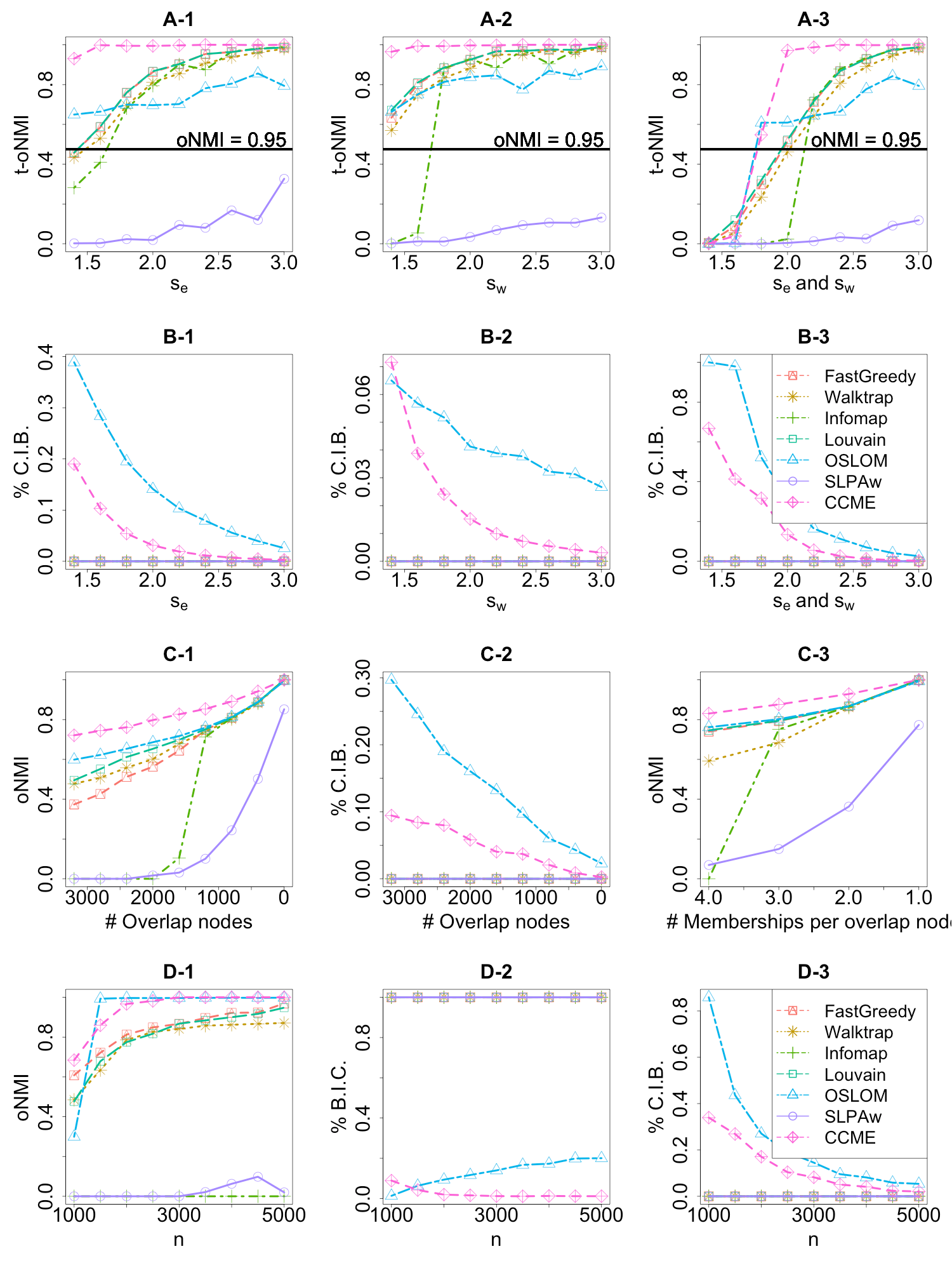}
	\caption{\label{fig:sbm} Simulation results described in Sections \ref{Simulations:disjoint}-\ref{Simulations:background}. Legends refer to all plots.}
\end{figure}

%% file: section_files/bw_Sec6_Data.tex
\section{Applications}\label{Data}

In this section, we discuss applications of CCME, OSLOM, and SLPAw (the methods capable of returning overlapping communities) to two real data sets.

\subsection{U.S.\ airport network data}\label{Data:airports}
The first application involves commercial airline flight data, obtained from the Bureau of Transportation Statistics (www.transtats.bts.gov). For each month from January to July of 2015, we created a weighted network with U.S. airports as nodes, edges connecting airports that exchanged flights, and edges weighted by aggregate passenger count. We also constructed a year-aggregated network, formed simply by taking the union of the month-wise edge sets, and adding the month-wise weights. In Figure \ref{fig:airports_years}, we display the methods' results when applied to the June and year-aggregated data sets from 2015. Each discovered community (within-method) has a unique color and shape. Each overlapping node is plotted multiple times, one for each community in which it was placed. For a clearer visualization of communities, background nodes are not shown.

Overall, the CCME results, in contrast to results from OSLOM and SLPAw, suggest that many airports in the U.S.\ airport system may not participate in meaningful community behavior. The fact that CCME performs multiple testing against an explicit null model gives this result some validity. Furthermore, airports in significant communities tend to be located near large hubs or in geographically isolated areas. We also see that, with the monthly data, OSLOM and CCME tended to find communities consistent with geography, whereas SLPAw placed most of the network into one community. With the year-aggregated data, OSLOM also agglomerated most airports, whereas CCME continued to respect the geography. Since the aggregated data is much more edge-dense, this suggests the performance of OSLOM and SLPA may suffer on weighted graphs with high or homogeneous edge-density, whereas CCME is able to detect proper community structure from the weights alone. This aligns with the simulation results described in Section \ref{Simulations:disjoint}.

\begin{figure}
	\centering
	\makebox{
		\includegraphics[scale = 0.16]{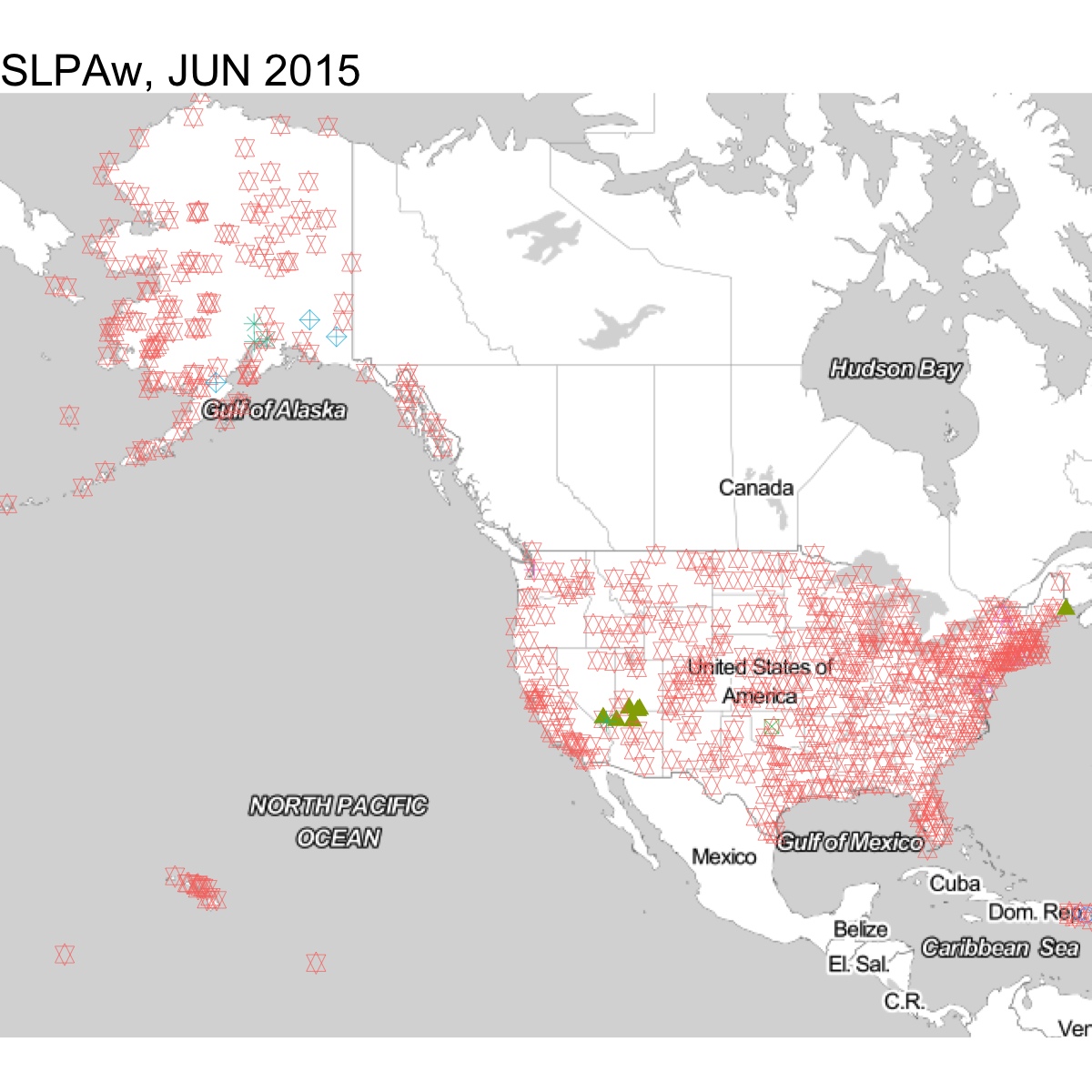}
		\includegraphics[scale = 0.16]{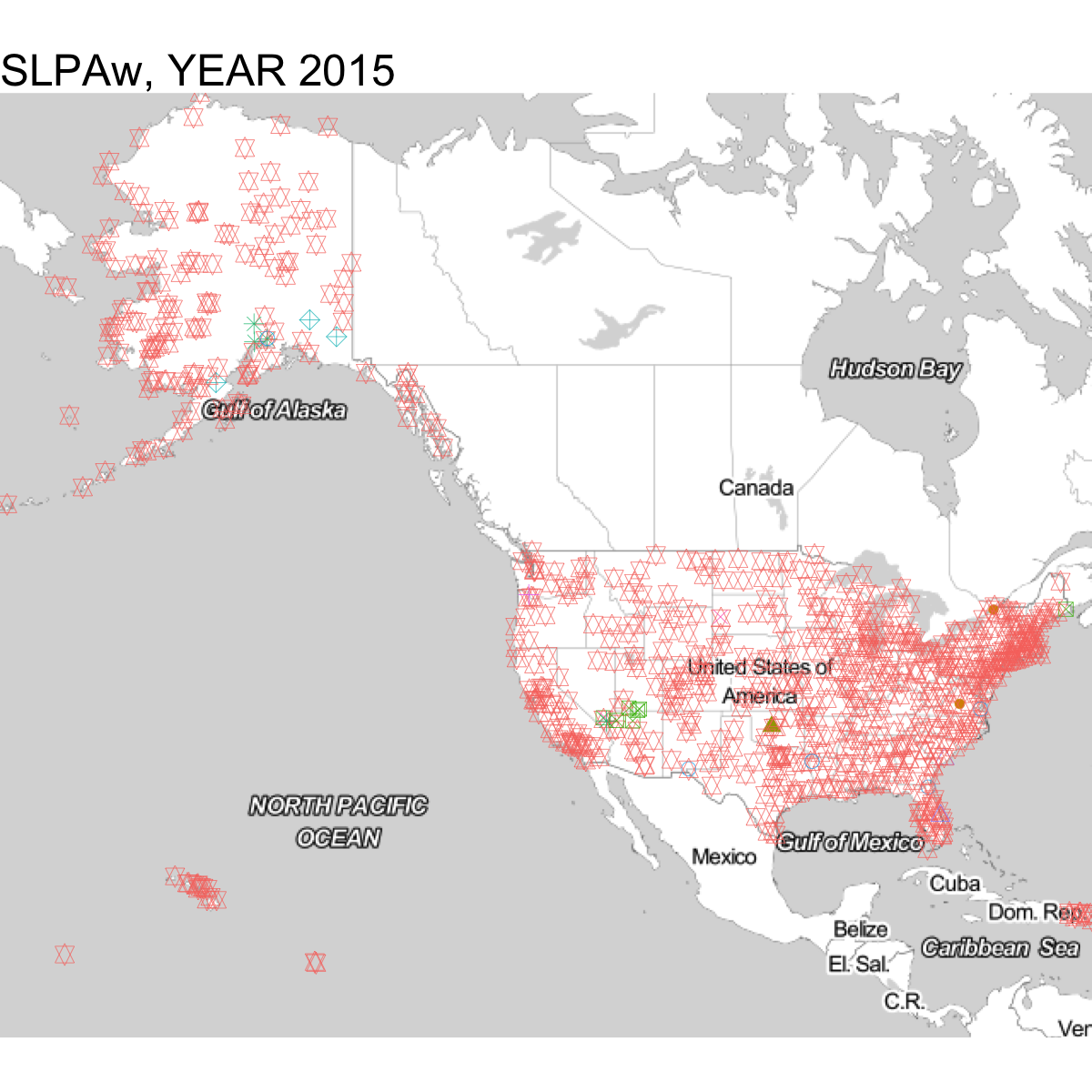}
	}
	\makebox{
		\includegraphics[scale = 0.16]{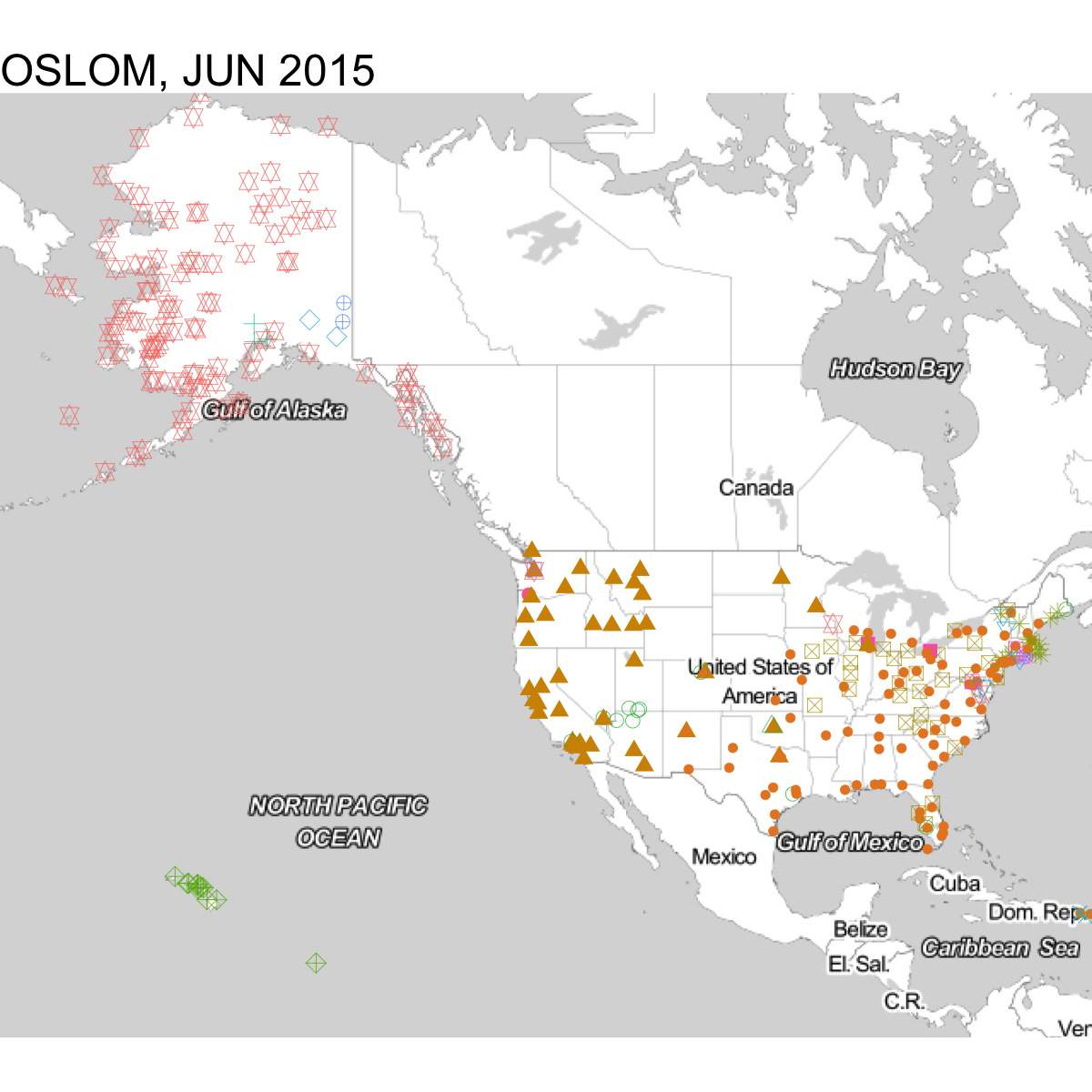}
		\includegraphics[scale = 0.16]{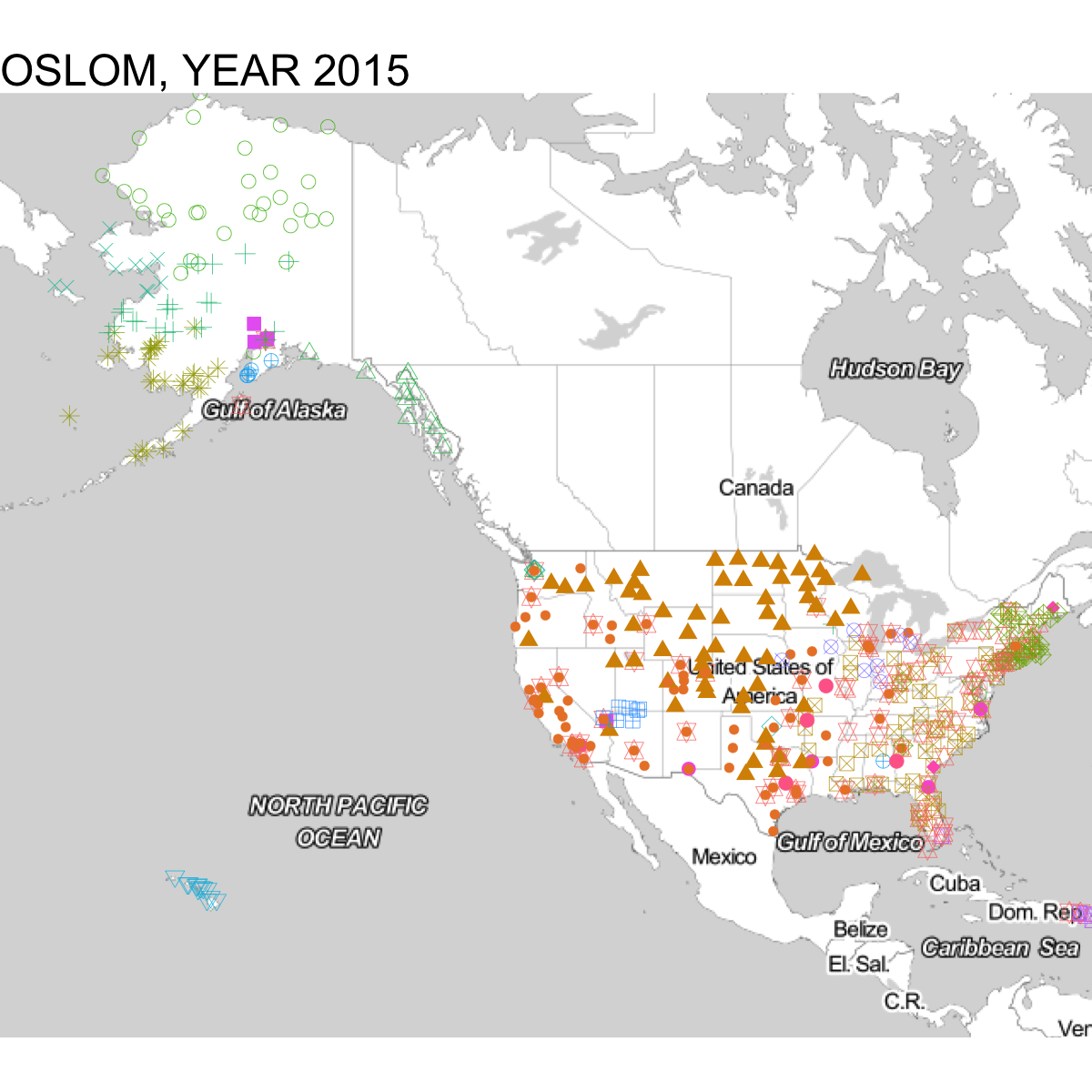}
	}
	\makebox{
		\includegraphics[scale = 0.16]{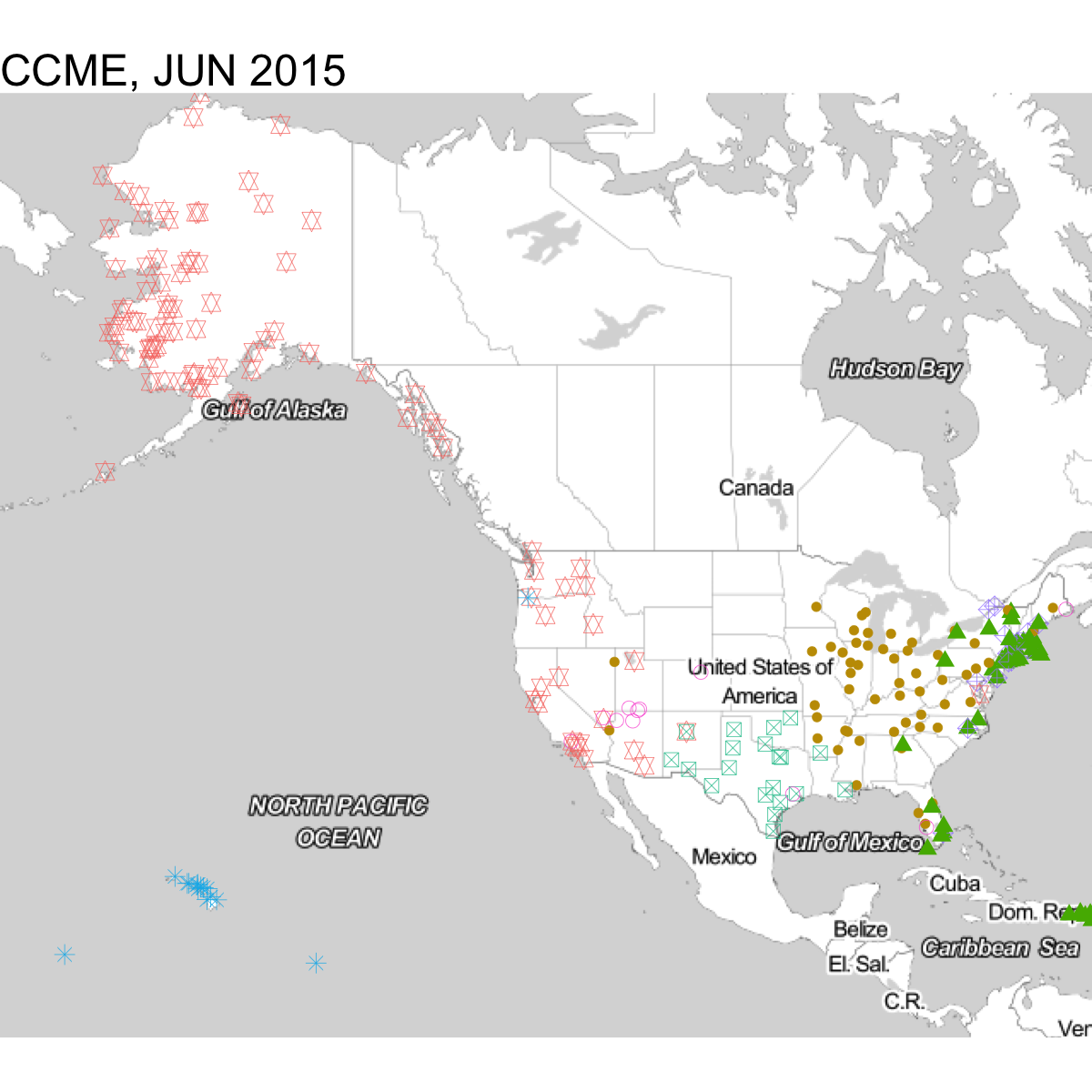}
		\includegraphics[scale = 0.16]{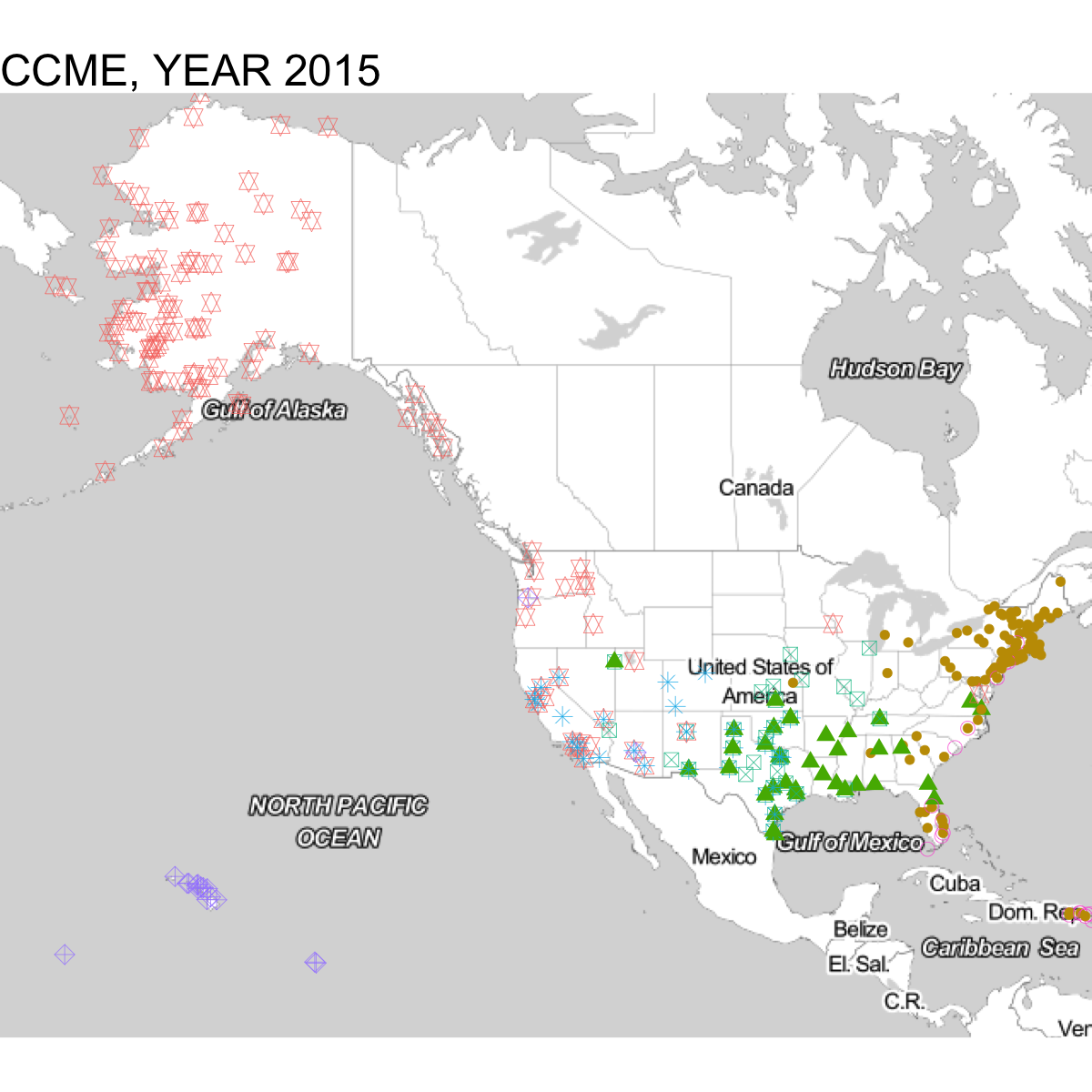}
	}
	\caption{\label{fig:airports_years}SLPAw, OSLOM, and CCME results from June 2015 and 2015-year-aggregated U.S. airport networks. Maps created with $\mathtt{ggmap}$ \citep{ggmap}.}
\end{figure}

\subsection{ENRON email network}\label{Data:enron}

An email corpus from the company ENRON was made available in 2009. The un-weighted network formed by linking communicating email addresses is well-studied; see www.cs.cmu.edu/\url{~}./enron for references and \cite{leskovec2010stanford} for the data. For the purposes of this paper, we derived a \emph{weighted} network from the original corpus, using message count between addresses as edge weights. Though the corpus was formed from email folders of 150 ENRON executives, we made the network from addresses found in \emph{any} message. This full network has 80,702 nodes, comprised of a majority of non-ENRON addresses, and likely many spam or irrelevant senders. Thus, the network has many potential ``true" background nodes. We applied CCME, OSLOM, and SLPAw to the network to see which methods best focused on 
company-specific areas of the data.

Tables \ref{table:email_stats1} and \ref{table:email_stats2} give basic summaries of the results, which show 
noticeable differences between the outputs of the methods.  CCME placed far fewer into nodes 
into communities, but detected larger communities with more overlapping nodes. 
Notably, CCME had the highest percentage of ENRON addresses among nodes it placed into communities 
(see Table \ref{table:domains}). These results suggest that CCME was more sensitive to 
critical relationships in the network.

\begin{table}[!htb]
\caption{\label{table:email_stats1} Metrics from methods' results on ENRON network: number of communities, minimum community size, median community size, maximum community size, count of nodes in any community}
\centering
\fbox{
\scriptsize
\begin{tabular}{*{6}{c}}
 & Num.Comms & Min.size & Med.size & Max.size & Num.Nodes \\ 
\hline
CCME & 185 & 2 & 687 & 5416 & 14552 \\ 
OSLOM & 405 & 2 & 19 & 770 & 17635 \\ 
SLPAw & 2138 & 2 & 4 & 4793 & 79316 \\
\end{tabular}}
\end{table}

\begin{table}[!htb]
\caption{\label{table:email_stats2} Metrics from methods' results on ENRON network: number of overlapping nodes, minimum \# of memberships, median \# of mem'ships, max.\ \# of mem'ships}
\centering
\fbox{
\scriptsize
\begin{tabular}{*{5}{c}}
 & Num.OL.Nodes & Min.mships & Med.mships & Max.mships \\ 
  \hline
CCME & 8104 & 2 & 9 & 78 \\ 
OSLOM & 462 & 2 & 2 & 8 \\ 
SLPAw & 3860 & 2 & 2 & 4 \\ 
\end{tabular}}
\end{table}

\begin{table}[!htb]

\caption{\label{table:domains} Top domains associated with community nodes from each method, by proportion}
\centering

\fbox{
\scriptsize
\begin{tabular}{*{2}{c}}
%  \hline
CCME.Domains & Prop. \\ 
  \hline
enron.com & 0.784 \\ 
aol.com & 0.008 \\ 
cpuc.ca.gov & 0.006 \\ 
pge.com & 0.004 \\ 
socalgas.com & 0.003 \\ 
dynegy.com & 0.003 \\ 
\end{tabular}}
% ------------------------------------------
%\quad
% ------------------------------------------
\fbox{
\scriptsize
\begin{tabular}{*{2}{c}}
%  \hline
OSLOM.Domains & Prop. \\ 
  \hline
enron.com & 0.529 \\ 
aol.com & 0.029 \\ 
haas.berkeley.edu & 0.016 \\ 
hotmail.com & 0.015 \\ 
yahoo.com & 0.009 \\ 
jmbm.com & 0.005 \\ 
\end{tabular}}
% ------------------------------------------
%\quad
% ------------------------------------------
\fbox{
\scriptsize
\begin{tabular}{*{2}{c}}
%  \hline
SLPAw.Domains & Prop. \\ 
  \hline
enron.com & 0.423 \\ 
aol.com & 0.039 \\ 
hotmail.com & 0.023 \\ 
yahoo.com & 0.016 \\ 
haas.berkeley.edu & 0.007 \\ 
msn.com & 0.006 \\
\end{tabular}}
% ------------------------------------------
\end{table}

%% file: section_files/bw_Sec7_Discussion.tex
\section{Discussion}\label{Discussion}
In this paper, we introduced the continuous configuration model, which is, to the best of our knowledge, the first null model for community detection on weighted networks. The explicit generative form of the null model allowed the specification of CCME, a community extraction method based on sequential significance testing.  We showed that a standardized statistic for the tests is asymptotically normal, a result which enables an analytic approximation to p-values used in the method.  We also proved asymptotic consistency under a weighted stochastic block model for the core algorithm of the method.

On simulated networks the proposed method CCME is competitive with commonly-used community detection methods. 
CCME was the dominant method for simulated networks with large numbers of overlapping nodes. 
Furthermore, on networks with background nodes, CCME was the only method to correctly 
label true background nodes while maintaining high detection power and accuracy for 
nodes belonging to communities. On real data, CCME gave results that were both interpretable and revelatory with respect to the natural system under study.

We expect that the continuous configuration model will have applications outside the setting of this paper, just as the binary configuration model has been studied in diverse contexts. One may investigate the distributional properties of many different graph-based statistics under the model, 
as a means of assessing statistical significance in practice. 
For instance, an appropriate theoretical analysis could yield an approach to the assessment of statistical significance of weighted modularity. Theorem \ref{thm:CLT} may be precedent for this endeavor. Another benefit of an explicit null for weighted networks is the potential for simulation. Using the continuous configuration model, and parts of the framework presented in this paper, one can generate weighted networks having true background nodes with arbitrary expected degree and strength distributions.

\subsection{Acknowledgements and Remarks} The authors thank Dr.\ Peter J.\ Mucha for helpful suggestions about the presentation and contextualization of this paper's contributions. The $\mathtt{R}$ code for the CCME method is available in the github repository `jpalowitch/CCME'. The code for reproducing the analyses in Sections \ref{Simulations} and \ref{Data} is  available at the github repository `jpalowitch/CCME\_analyses'.

%% file: supplemental1_section_files/CLT_NA.tex
\section{Proof of Proposition \ref{sub-mean-sd}}\label{Prop1proof}

	Equation \ref{eq:mean_fun} follows immediately from the observation in equation \ref{Model:expv1} and the definition of $r_{uv}(\vs)$. Next, note that 
	\[
	\expv\left(W_{uv}|A_{uv}\right) = f_{uv}(\vd,\vs)A_{uv},
	\;\text{ and }\;
	\text{Var}(W_{uv}|A_{uv}) = \vpar f_{uv}(\vd,\vs)^2A_{uv}.
	\]
	Thus, using the law of total variance,
	\begin{align*}
	\text{Var}(W_{uv}) & = f_{uv}(\vd, \vs)^2\text{Var}(A_{uv}) + \vpar f_{uv}(\vd, \vs)^2\expv(A_{uv})\\
	& = f_{uv}(\vd, \vs)^2\ruvdag(\vd)(1 - \ruvdag(\vd)) + \vpar f_{uv}(\vd, \vs)^2\ruvdag(\vd)\\
	& = r_{uv}(\vs)f_{uv}(\vd,\vs)\left(1 - \ruvdag(\vd) + \vpar\right)
	\end{align*}
	Summing over $v\in B$ gives equation \ref{eq:var_fun}.\qed

\section{Proof of Theorem \ref{thm:CLT} and supporting lemmas.}\label{clt-proof}

Here we give the proof of Theorem \ref{thm:CLT} in Section \ref{Method:CLT}. We start with supporting lemmas.  Recall the definition of the average degree parameter 
$\lambda_n$, 
the normalized $r^{\text{th}}$-moment $L_{n,r}$, 
and other associated definitions from Section \ref{Method:CLT}. For the purposes of the results below, we define the following generalization of $L_{n,r}$, given a node set $B_n\subseteq N$ with $b_n := |B_n|$:
\[L_{n,r}(B_n) := b_n^{-1}\underset{u\in B_n}{\sum}\{d_n(u)/\lambda_n\}^r\]
Note that $L_{n, r}(N) = L_{n,r}$. Recall that in the setting of Theorem \ref{thm:CLT}, the node set $B_n$ is chosen uniformly from the node set $N$. The first result involves a \emph{deterministic} sequence $\{B_n\}_{n\geq1}$:

% Deterministic CLT proposition
\begin{lemma}\label{prop:deterministic_CLT} 
For each $n>1$, let $\Net_n$ be generated by the continuous configuration model with parameters $\theta_n = (\vd_n, \vs_n, \vpar_n)$ and common weight distribution $F$. Fix a node sequence $\{u_n\}_{n>1}$ with $u_n\in N$ and a positive integer sequence $\{b_n\}_{n>1}$ with $b_n\leq n$. Suppose the parameter sequence $\{d_n(u_n)\}_{n\geq1}$ satisfies
\[
\frac{d_n(u_n)b_n}{n}\rightarrow\infty\text{ as }n\rightarrow\infty
\]
Fix $\eps > 0$ as in Assumption $\ref{assumption:moment}$, and choose $\delta\in(0,1)$ such that $2\beta\delta < \eps$. Fix a sequence of sets $\{B_n\}_{n>1}$ with $|B_n| = b_n$ for all $n$, and suppose that for $r = 2\beta + 1$ and $r = \beta(2+\delta) + 1$, the sequence $\{L_{n, r}(B_n)\}_{n>1}$ is bounded away from zero and infinity. Then
\[
\frac{S(u_n,B_n,\Net_n) - \mu_n(u_n,B_n|\theta_n)}{\sigma_n(u_n,B_n|\theta_n)}\Rightarrow\mathcal{N}(0,1)\;\text{ as }\;n\rightarrow\infty
\]
\end{lemma}

% Deterministic CLT Proof
\noindent
\begin{proof} In what follows, the functions $r_{uv}$ and $\tilde r_{uv}$ from Section \ref{Notation} will be used extensively. Note that for any nodes $u,v$, $\expv W_{uv} = r_{uv}(\vs)$. Thus by the classical Lyapunov central limit theorem it suffices to show that
\begin{equation}\label{eq:lyapunov}
\frac{\underset{v\in B_n}{\sum}\expv|W_{u_n,v} - r_{u_nv}(\vs_n)|^{2+\delta}}
{\left(\sqrt{\underset{v\in B_n}{\sum}\expv\left\{(W_{u_n,v} - r_{u_nv}(\vs_n))^2\right\}}\right)^{2+\delta}}\rightarrow0
\end{equation}
as $n$ tends to infinity. The following derivations hold for any fixed $n>1$, so we suppress dependence on $n$ from $u_n$, and $B_n$, and similar expressions. For the numerator of \eqref{eq:lyapunov}, we have
\begin{align}
\expv|W_{u,v} - \rfun{\vs}|^{2+\delta} & = \left(\frac{\rfun{ \vs }}{\ruvdag(\vd )}\right)^{2+\delta}\expv\left(|\xi_{uv}A_{uv} - \ruvdag(\vd)|^{2+\delta}\right)\nonumber\\
& = \ffun{ \vd }{ \vs }^{2+\delta}\cdot\expv\left(|\xi_{uv}A_{uv} - \ruvdag(\vd)|^{2+\delta}\right),\label{eq:num0}
\end{align}
by definition of the model in Section $\ref{Model:statement}$.  Moreover, by the law of total variance,
\begin{align}
\expv(|\xi_{uv}A_{uv} - \ruvdag(\vd)|^{2+\delta}) \;=\;& (1 - \ruvdag(\vd))\ruvdag(\vd)^{2+\delta} + \ruvdag(\vd)\expv|\xi_{uv} - \ruvdag(\vd)|^{2+\delta}\nonumber\\
=\;&\left\{(1 - \ruvdag(\vd))\ruvdag(\vd)^{1+\delta} + \expv|\xi_{uv} - \ruvdag(\vd)|^{2+\delta}\right\}\cdot \ruvdag(\vd)\nonumber\\
\leq\;&C\cdot \ruvdag(\vd)\label{eq:num1}
\end{align}
for some positive constant $C$, by Assumption \ref{assumption:distribution}. 
Next, we note that by Assumption \ref{assumption:power}, there exist positive constants $a < c$ such that for all $v\in N$,
\[a\cdot d_n(v)^\beta\leq\frac{s_n(v)}{d_n(v)}\leq c\cdot d_n(v)^\beta,\]
for $n$ sufficiently large. Thus, if $r_{uv}(\vd)\leq 1$, $\ruvdag(\vd) = \rfun{ \vd }$, and 
\begin{equation}
\ffun{ \vd }{ \vs } = \frac{\rfun{ \vs }}{\ruvdag( \vd )} = \left(\frac{d_T}{s_T}\right)\frac{s(u)s(v)}{d(u)d(v)}\leq c\cdot \left(\frac{d_T}{s_T}\right)\{d(u)d(v)\}^\beta.\label{eq:num2}
\end{equation}
If $r_{uv}(\vd)>1$, $\ruvdag(\vd) = 1$, and by Assumption \ref{assumption:ruv} there exists $c'$ such that
\begin{align}
\ffun{ \vd }{ \vs } = \frac{s(u)s(v)}{s_T} & \leq c\cdot \left(\frac{d(u)d(v)}{s_T}\right)\{d(u)d(v)\}^\beta \nonumber\\
& = c\cdot \left(\frac{d_T}{s_T}\right)r_{uv}(\vd)\{d(u)d(v)\}^\beta\leq c'\cdot \left(\frac{d_T}{s_T}\right)\{d(u)d(v)\}^\beta.\label{eq:num2-2}
\end{align}
Therefore, combining \eqref{eq:num1}-\eqref{eq:num2-2} with \eqref{eq:num0}, there exists $C>0$ such that
\begin{align}
\expv|W_{u,v} - r_{uv}(\vs)|^{2+\delta} \; \leq& \; C\left(\frac{d_T}{s_T}\right)^{2 + \delta}\cdot\{d(u)d(v)\}^{\beta(2+\delta)}\ruvdag(\vd)\nonumber\\
=& \; C\left(\frac{d_T}{s_T}\right)^{2 + \delta}\cdot\{d(u)d(v)\}^{\beta(2+\delta)}\frac{d(u)d(v)}{d_T}\nonumber\\
\leq & \; C\cdot d_T^{1+\delta}s_T^{-(2 + \delta)}\cdot\{d(u)d(v)\}^{\beta(2+\delta) + 1}\label{eq:num_final}
\end{align}
A similar analysis of the summands in the denominator of \eqref{eq:lyapunov} gives
\begin{equation}\label{eq:den_final}
\expv\left\{(W_{u,v} - r_{uv}(\vs))^2\right\} \geq C'\cdot d_Ts_T^{-2}\cdot\{d(u)d(v)\}^{2\beta + 1}
\end{equation}
for appropriately chosen $C'$. Let $b = |B|$. Combining \eqref{eq:num_final} and \eqref{eq:den_final}, with some algebra, we find that the 
left side of \eqref{eq:lyapunov} is (up to a constant) less than
\begin{align}
&
\left(\frac{d(u)}{d_T}\right)^{-\delta/2}
\cdot
\frac{\underset{v\in B}{\sum}d(v)^{\beta(2+\delta) + 1}}{\left(\underset{v\in B}{\sum}d(v)^{2\beta + 1}\right)^{1 + \delta / 2}}\nonumber\\[.1in] % end of line1
\;=\;&
\left(\frac{d(u)}{d_T}b\lambda\right)^{-\delta/2}
\cdot
\frac{b^{-1}\underset{v\in B}{\sum}\left(d(u)/\lambda\right)^{\beta(2+\delta) + 1}}{\left\{b^{-1}\underset{v\in B}{\sum}\left(d(u)/\lambda\right)^{2\beta + 1}\right\}^{1 + \delta/2}}\nonumber\\[.1in] % end of line2
\;=\;&
\left(\frac{d(u)}{d_T}b\lambda\right)^{-\delta/2}
\cdot
\frac{L_{n, \beta(2+\delta) + 1}(B)}{\left(L_{n, 2\beta + 1}(B)\right)^{1 + \delta/2}}
\ = \ 
O\left\{\left(\frac{d(u)}{d_T}b\lambda\right)^{-\delta/2}
\right\}\label{eq:last_one}
% end of line3
\end{align}
where the final term follows from our assumptions on $L_{n, \beta(2+\delta) + 1}(B_n)$ 
and $L_{n, 2\beta + 1}(B_n)$. 
By definition, $d_{n, T} = n\lambda_n$, so the final expression above is 
$O\left\{\left(d_n(u_n)b_n/n\right)^{-\delta/2}\right\} = o(1)$ by assumption. 
Thus \eqref{eq:lyapunov} holds and the result follows.
\end{proof}

\vskip.1in
\noindent 
We now proceed with the proof of Theorem \ref{thm:CLT}. 
%The main idea is as follows. Note that 
Proposition \ref{prop:deterministic_CLT} yields the CLT for $S(u_n,B_n,\Net_n)$ for a deterministic sequence 
of vertex sets $\{B_n\}_{n\geq 1}$ satisfying regularity properties. The remainder of the argument 
shows that if $B_n$ is selected uniformly at random then, under the assumptions of Theorem \ref{thm:CLT}, 
these regularity properties are satisfied with high probability. 
%This is needed to complete the proof (Section \ref{proof_completion}).
We begin with a few preliminary definitions and results.
% Here we give a definition and related result from \cite{van2000asymptotic}:

% Van der Vaart Theorem
\begin{definition}\label{def:aui} 
A sequence of random variables $\{X_n\}_{n\geq1}$ is said to be \emph{asymptotically uniformly integrable} if
\[\underset{M\rightarrow\infty}{\lim}\underset{n\rightarrow\infty}{\limsup}\;\expv\left\{|X_n|\IND(|X_n|>M)\right\} = 0\]
\end{definition}

\begin{theorem}\label{thm:vdv} 
Let $f:\reals^k\mapsto\reals^k$ be measurable and continuous at every point in a set $C$. 
Suppose $X_n\xrightarrow{w}X$ where $X$ takes its values in an interval $C$. Then $\expv f(X_n)\rightarrow \expv f(X)$ if and only if the sequence of random variables $f(X_n)$ is asymptotically uniformly integrable.
\end{theorem}

\begin{proof} See Asymptotic Statistics (Van der Vaart 2000), page 17.
\end{proof}

\noindent We now give a technical lemma (needed for a subsequent result) which uses Theorem \ref{thm:vdv}:

% Moment lemma
\begin{lemma}\label{lemma:moment}
Let $X_1, X_2, \ldots$ be non-negative random variables and let $s, \eps >0$. 
If the sequences $\{\expv X_n^s\}_{n \geq 1}$ and $\{\expv X_n^{s + \eps}\}_{n \geq 1}$ are 
bounded away from zero and infinity, then $\{\expv X_n^r\}_{n\geq1}$ is bounded away
from zero and infinity for every $r\in (0, s + \eps)$.
\end{lemma}

\begin{proof}
%First we show that $\underset{n\rightarrow\infty}{\liminf}\;\expv X_n^r$ is non-zero for all $r\in(0, s+\eps)$, by way of contradiction. 
Suppose by way of contradiction that there exists $t \in(0, s+\eps)$ such that $\liminf_n \expv X_n^{t} = 0$. 
Then $\lim_k \expv X_{n_k}^{t} =0$ along a subsequence $\{n_k\}$.  As the random variables
$X_{n_k}^{t}$ are non-negative, $X_{n_k}^{t} \xrightarrow{d}0$, and it follows from the 
continuous mapping theorem that $X_{n_k} \xrightarrow{w} 0$. 
As $M^{\eps/s} \, X_n^s \, \IND(X_n^s > M) \leq X_n^{s+\eps}$, we find that
\[
\underset{M \rightarrow \infty}{\lim}\underset{k\rightarrow\infty}{\limsup}\;\expv\{X_{n_k}^s \IND(X_{n_k}^s > M)\}
\ \leq \  % is less than or equal to
\underset{M \rightarrow \infty}{\lim}M^{-\eps/s} \ 
\underset{k\rightarrow\infty}{\limsup}\;\expv(X_{n_k}^{s+\eps}) = 0
\]
as $\expv(X_n^{s+\eps})$ is bounded by assumption.  It then follows from Theorem $\ref{thm:vdv}$ and
the fact that $X_{n_k}^s \xrightarrow{w} 0$ that $\expv X_{n_k}^s \rightarrow 0$ as $k \rightarrow\infty$,
violating our assumption that $\expv X_n^s$ is bounded away from zero.  We conclude that 
$\expv X_n^r$ is bounded away from zero for $r \in (0,s+\eps)$.
On the other hand, if $r \in (0, s+\eps)$ then for each $n \geq 1$
\[
\expv\{X_n^r \IND(X_n > 1)\}
\ \leq \ 
\expv\{X_n^{s+\eps} \IND(X_n>1)\}
\ \leq \ 
\sup_n \expv\{X_n^{s+\eps}\}
\]
As the last term is finite by assumption and $\expv\{X_n^r \IND(X_n \leq 1)\}$ is at most one, 
it follows that $\expv(X_n^r)$ is bounded.
\end{proof}

% U.A.R. draws lemma
\begin{lemma}\label{lem:uarDraws} 
Suppose a degree parameter sequence $\{{\bf d}_n\}_{n\geq1}$ satisfies Assumption \ref{assumption:moment} 
from Section \ref{Method:CLT}. For each $n$, let $B_n$ be a randomly chosen subset of $N$ of size $b_n$, 
where $b_n \rightarrow \infty$. Fix $\eps>0$ as in Assumption $\ref{assumption:moment}$, 
and choose $\delta$ so that $2\beta\delta < \eps$. 
Then for every $r \in (0,\;\beta(2+\delta)+1]$, there exists an interval $I_r = (a_r,b_r)$ with $0 <a_r < b_r < \infty$ 
such that
$
\prob\{L_{n, r}(B_n) \in I_r\} \rightarrow 1 \text{ as } n \rightarrow \infty 
$.
\end{lemma}

\vskip.1in

\noindent
{\bf Remark:} Note that the function $L_{n,r} (\cdot)$ is non-random.  The probability appearing in the conclusion 
of Lemma \ref{lem:uarDraws} depends only on the random choice of the vertex set $B_n$.

% U.A.R. lemma Proof
\begin{proof}
Let $D_n$ and $D_n'$ be drawn uniformly-at-random from ${\bf d}_n$ without replacement, and 
fix $r \in (0,\;\beta(2+\delta)]$.   A routine calculation gives 
\[
\text{Var}\{L_{n, r}(B_n)\} = b_n^{-1}\lambda_n^{-2r}\left[\text{Var}\{D_n^r\} + \{b_n - 1\}\text{Cov}\{D_n^r, (D_n')^r\}\right].
\]
Note that $\expv(D_n^r) = \lambda_n^rL_{n, r}$ and $\expv(D_n^{2r}) = \lambda_n^{2r}L_{n, 2r}$, so
$\text{Var}(D_n^r) = \lambda_n^{2r}(L_{n,2r} - L_{n,r})$. 
Furthermore, a simple calculation shows that $\text{Cov}\{D_n^r, (D_n')^r\}$ is negative for every $r$, and
therefore $\text{Var}\{L_{n, r}(B_n)\} \leq b_n^{-1}(L_{n, 2r} - L_{n, r})$. 
%ABN: Stopped here
Our choice of $\delta$ ensures that $2r < 4 \beta + 2 + \eps$, and it then follows from
Lemma \ref{lemma:moment} and Assumption \ref{assumption:moment} that 
$L_{n, 2r}$ and $L_{n, r}$ are bounded.  Thus $\text{Var}\{L_{n, r}(B_n)\} = O(b_n^{-1})$. 
Define $\Delta := \liminf_n \; L_{n, r} / 2$, which is positive by Assumption \ref{assumption:moment}, and let
\begin{equation}\label{eq:I_r}
I_r := \left(\underset{n\rightarrow\infty}{\liminf}\;L_{n, r}-\Delta,\;\underset{n\rightarrow\infty}{\limsup}\;L_{n, r} + \Delta\right)
\end{equation}
As $\expv\{L_{n, r}(B_n)\} = L_{n, r}$, an application of Chebyshev's inequality yields the bound
\begin{align*}
\prob\{L_{n, r}(B_n)\notin I_r\}
\ \leq \ &\;\prob\{|L_{n, r}(B_n) - \expv[L_{n, r}(B_n)]| > \Delta / 2\}\\[.06in]  
\ \leq \ &\;\frac{4 \text{Var}\{L_{n, r}(B_n)\}}{\Delta^2}
\ = \ O(b_n^{-1}).
\end{align*}
As $b_n$ tends to infinity with $n$, the result follows.
\end{proof}

\subsection{Completing the proof of Theorem \ref{thm:CLT}.}\label{proof_completion}
Let $\eps$ and $\delta$ be as in Proposition \ref{prop:deterministic_CLT} and Lemma \ref{lem:uarDraws}. Note that since $d_n(u_n)\leq n$ for all $n$, our assumption that $b_nd_n(u_n)/n\rightarrow\infty$ implies $|B_n| = b_n\rightarrow\infty$. Hence by lemma \ref{lem:uarDraws}, we have that for both $r = \beta(2 + \delta) + 1$ and $r = 2\beta + 1$, there exists a positive, finite interval $I_r$ such that $\prob\{L_{n, r}(B_n)\in I_r\}\rightarrow1$ as $n\rightarrow\infty$. Thus given any subsequence $\set{n_k}_{k\geq 1}$ we can find a further subsequence $\set{n_k^{\prime}}_{k\geq 1}$ such that $L_{n_k^{\prime}, r}(B_{n_k^{\prime}})\in I_r$ almost surely as $k\rightarrow\infty$, which means this sequence is bounded away from zero and infinity in $k$. Now using Proposition \ref{prop:deterministic_CLT}, for almost every $\omega$ we have 
\begin{equation}\label{eq:CLT_ss}
\frac{S_{n'_k}(u_{n_{k}^{\prime}},B_{n^{\prime}_k},\Net_{{n^{\prime}_k}}) - \mu_{n_k'}(u_{n_{k}^{\prime}},B_{n^{\prime}_k}|\theta_{n^{\prime}_k})}{\sigma_{n_k'}(u_{n_{k}^{\prime}},B_{n^{\prime}_k}|\theta_{n^{\prime}_k})}\Rightarrow\mathcal{N}(0,1)\;\text{ as }\;k\rightarrow\infty\nonumber
\end{equation}
Applying the subsequence principle completes the proof. \qed

%% file: supplemental1_section_files/Consistency.tex
\section{Proof of Theorems \ref{initial-to-community}-\ref{thm:consistency} and supporting lemmas.}\label{consistency-proofs}

%%%%%%%%%
%%% Expectations Lemma
%%%%%%%%%

Throughout this section, notation and conventions from Section \ref{WSBM} will be used, though we suppress dependence on $n$ for convenience. Further recall functions $r$ and $f$ from Section \ref{Notation}. The following additional notation will be used throughout this section:
\begin{itemize}
	\item Define $\phi_T := \sum_{v\in N}\phi(v)$ and $\psi_T := \sum_{v\in N}\psi(v)$. For each $K\geq j \geq 1$, define $\tilde \pi^0_j := \sum_{v\in\cC_j}\phi(v)/\phi_T$ and $\tilde \pi_j := \sum_{v\in\cC_j}\psi(v)/\psi_T$. Let $\tilde \vpi^0$ and $\tilde\vpi$ be the associated vectors.
	\item Let $\bra\cdot,\cdot\ket$ denote the vector dot-product. For a general symmetric matrix $\mathbf{A}$, let $\mathbf{A}_{ij}$ be the $i,j$-th entry, and $\mathbf{A}_i$ the $i$-th column. Define $\mH := \mP\cdot\mM$, the entry-wise product.
	\item Let $D(u),S(u)$ be the random degree, strength of node $u\in N$, let $\tilde d(u)$, $\tilde s(u)$ be the corresponding expectations, and let $\vD,\vS,\bar\vd,\bar\vs$ be the associated $n$-vectors. Define $\bar s_T := \sum_{v\in N}\bar s(v)$ and $\bar d_T := \sum_{v\in N}\bar d(v)$.
\end{itemize}
We now define a \emph{empirical} population version of the variance estimate:
\begin{definition}\label{pop-var}
	Fix $n>1$ and let $A$ and $W$ be the edge and weight matrices from $\Net_n$, the $n$-th random weighted network from the sequence in the setting of Theorem \ref{initial-to-community}. Let $\vx,\vy$ be arbitrary $n$-vectors with positive entries. For nodes $u,v\in N$, define
	\begin{equation*}
	V_{uv}(\vx,\vy):=\left(W_{uv} - \ffun{\vx}{\vy}\right)^2,
	\;\;\;\;\;
	v_{uv}(\vx,\vy) := \expv\left\{V_{uv}(\vx,\vy)\big|A_{uv} = 1\right\}.
	\end{equation*}
	Define the \emph{empirical} population variance estimator as follows:
	\begin{equation*}
	\vpar_\ast(\vx,\vy) := \frac{\sum_{u,v:A_{uv} = 1} v_{uv}(\vx,\vy)}{\sum_{u,v:A_{uv} = 1}f_{uv}(\vx,\vy)^2}
	\end{equation*}
\end{definition}
The estimator $\kappa_\ast(\vx,\vy)$ is called ``empirical" because it depends on the random edge set $E$. Despite this, it has a deterministic bound, a fact which is part of Lemma \ref{expv-lemma}. Throughout the remaining results, denote $\Theta := (\vD,\vS,\hat \kappa(\vD,\vS))$ and $\theta_\ast := (\bar \vd, \bar \vs, \kappa_\ast(\bar \vd, \bar \vs))$, 
where the estimator $\hat\kappa$ is the estimator from Section \ref{Model:specification}. 

Recall the definition of the asymptotic order of the average degree $\lambda_n:=n\rho_n$, from Section \ref{Consistency-theorem} in the main text. With this and the conventions above, Lemma \ref{expv-lemma} establishes basic facts about the WSBM:
\begin{lemma}\label{expv-lemma}
	Fix $n>1$, and  let $\Net_n$ be a random network generated by a WSBM. For all nodes $u,v\in N$, under Assumptions \ref{bounded-parameter-assumption} and \ref{bounded-weight-assumption},
	
	\begin{enumerate}[(1)]
		\item $\bar d(u) = \lambda_n\phi(u)\bra \tilde\vpi^0,\mP[c(u)]\ket$ and $\bar s(u) = \lambda_n\psi(u)\bra \tilde\vpi,\mH[c(u)] \ket$
		\item $m_-^2\leq\bar d(u)/\lambda_n\leq m_+^2$ and $m_-^3\leq \bar s(u)/\lambda_n\leq m_+^3$
		\item $m_-\leq\bar d_T/n\lambda_n\leq m_+$ and $m_-^2\leq \bar s_T/n\lambda_n\leq m_+^2$
		\item $m_-^4/m_+^1\leq r_{uv}(\bar \vd)/\rho_n\leq m_+^4/m_-^1$ and $m_-^6/m_+^2\leq r_{uv}(\bar \vs)/\rho_n \leq m_+^6/m_-^2$
		\item $m_-^2/m_+^2\leq \ffun{\phi}{\psi}\leq m_+^2/m_-^2$ and $m_-^{10}/m_+^3\leq \ffun{\bar \vd}{\bar \vs}\leq m_+^{10}/m_-^{3}$
		\item $0\leq V_{uv}(\bar \vd, \bar \vs)\leq (\eta m_+^2/m_-^2 + m_+^{10}/m_-^{3})^2$
		\item $0\leq \kappa_\ast(\bar \vd, \bar \vs)\leq g(\eta,m_-,m_+)$ where $g$ is a  deterministic function.
		\item{ There exist global constants $0<m_1<m_2<\infty$ independent of $n$ such that for any node set $B\subseteq N$,
		\[
		m_1|B|\rho_n \; 
		\leq\;
		\mu(u,B| \bar\vs),
		\;
		\sigma(u,B|\theta_\ast)^2
		\;
		\leq
		\;
		m_2|B|\rho_n
		\]
		}
	\end{enumerate}
\end{lemma}

\begin{proof}
	For (1), we have
	\begin{align*}
	\bar s(u) := \expv S(u) & = \sum_{j=1}^K\sum_{v\in \cC_j}\expv W_{uv}= \sum_{j=1}^K\sum_{v\in \cC_j}\rho_nr_{uv}(\phi)\mH_{c(u)j}\\
	& = \rho_n\sum_{j=1}^K\phi(u)n\tilde\pi_j\mH_{c(u)j}= \lambda_n\phi(u)\bra \tilde\vpi, \mH_{c(u)}\ket
	\end{align*}
	An identical calculation yields the expression for $\bar d(u)$. The inequalities in (2) then follow from Assumption \ref{bounded-parameter-assumption}. For (3), we again apply Assumption \ref{bounded-parameter-assumption} to the equation
	\[
	\bar s_T = \sum_{i=1}^K\sum_{v\in\cC_i}\bar s(u) = \sum_{i=1}^Kn\lambda_n\phi(u)\bra \tilde\vpi,\mH_i\ket = n\lambda_n\tilde\vpi^T\mH\tilde\vpi
	\]
	An identical equation yields the inequality for $\bar d_T$. (2) and (3) directly yield the inequalities in (4). Note that Assumption \ref{bounded-parameter-assumption} implies $m_-^2\leq nr_{uv}(\phi),nr_{uv}(\psi)\leq m_+^2$, which yields the first inequality of (5). The second inequality of (5) follows from (4). For part (6), note that by Assumption \ref{bounded-weight-assumption} and the first inequality in (5), we have
	\begin{equation}\label{Wuv-ineq}
	W_{uv} := f_{uv}(\phi,\psi)\xi_{uv}\leq(m_+^2/m_-^2)\eta
	\end{equation}
	The second inequality in (5) then yields (6). For part (7), recalling the definition of $\kappa_\ast(\bar \vd, \bar \vs)$ from Definition \ref{pop-var}, note first that, by (6), $0\leq v_{uv}(\bar \vd, \bar \vs)\leq(\eta m_+^2/m_-^2 + m_+^{10}/m_-^{3})^2$. Thus, by the second inequality (5),
	\[
	0\leq \kappa_\ast(\bar \vd, \bar \vs) := \frac{ \sum_{u,v:A_{uv} = 1} v_{uv}(\bar \vd, \bar \vs) }{ \sum_{u,v:A_{uv} = 1} \ffun{ \bar \vd}{\bar \vs}} \leq \frac{(\eta m_+^2/m_-^2 + m_+^{10}/m_-^{3})^2}{m_-^{10}/m_+^3}
	\] 
	For part (8), recall that
	\[
	\mu(u,B|\bar \vs) := \sum_{v\in B}r_{uv}(\bar \vs)
	\]
	The first inequality in (8) follows from applying the second inequality in (4). Similarly,
	\[
	\sigma(u,B|\theta_\ast)^2 := \sum_{v\in B}r_{uv}(\bar \vs)\ffun{\bar \vd}{\bar \vs}(1 - \ruvdag(\bar \vd) + \kappa_\ast(\bar \vd, \bar \vs))
	\]
	The second inequality in part (8) follows from parts (4), (5), and (7).
\end{proof}

%%%%%%%%%
%%% Gross Lemma 
%%%%%%%%%

The next lemma shows that, if the degrees and strengths of $\Net_n$ are bounded around their expected values, the empirical estimate of variance is bounded around the conditional population estimate, and the coefficient of variation of $S_n(u,B)$ is bounded around its population value. Define $D_T := \sum_{u\in N} D(u)$ as the (random) total degree. Recall that $\lambda_n$ is the asymptotic order the average of the \emph{expected} degrees $\bar d_T$.

\begin{lemma}\label{gross-lemma}
Fix $n>1$. Suppose Assumption \ref{bounded-parameter-assumption} holds. Define
\begin{equation}\label{deg-str-assump}
M(\vD,\vS) := \max_{u\in N}\left\{|S(u) - \bar s(u)|,|D(u) - \bar d(u)|\right\}.
\end{equation}
Then the following statements hold:
\begin{enumerate}[(1)]
\item{There exists small enough $t>0$ such that if $M(\vD,\vS)\leq \lambda_nt$,
\[
\big|\hat\kappa(\vD,\vS) - \kappa_\ast(\bar\vd,\bar\vs)\big| = \left|\frac{\sum_{u,v:A_{uv} = 1} V_{uv}(\bar \vd,\bar \vs) - v_{uv}(\bar\vd,\bar\vs)}{\sum_{u,v:A_{uv} = 1}f_{uv}(\bar\vd,\bar\vs)^2 + D_T\rho_nO(t)}\right| + \rho_nO(t)
\]
}
\item{Fix a constant $\eps>0$ independent of $n$. Assume $|\hat\kappa(\vD,\vS) - \kappa_\ast(\bar\vd,\bar\vs)|\leq \eps$. Then then there exists small enough $t>0$ (not depending on $\eps$) such that if $M(\vD,\vS)\leq t$, for all $B\subseteq N$, we have
\[
\left|\frac{\mu(u,B|\Theta)}{\sigma(u,B|\Theta)} - \frac{\mu(u,B|\theta_\ast)}{\sigma(u,B|\theta_\ast)}\right| = \sqrt{|B|\rho_n}O(t)
\]
}
\end{enumerate}
\end{lemma}

%---

\begin{proof}
$M(\vD,\vS)\leq \lambda_nt$ implies there exists a $n$-vector $\mathbf{a}_t$ with components in the interval $[-1,1]$ such that $S(u) = \bar s(u) + \lambda_nta_t(u)$. Therefore, defining $\bar a_t := n^{-1}\sum_v a_t(v)$,
\begin{align*}
r_{uv}(\vS) - r_{uv}(\bar\vs) & = \frac{\{\bar s(u) + \lambda_na_t(u)t\}\{\bar s(v) + \lambda_na_t(v)t\}}{\bar s_T + n\lambda_n\bar a_tt} - \frac{\bar s(u) \bar s(v)}{\bar s_T}\\
& = \frac{\bar s_T\{\bar s(u)a_t(v) + \bar s(v)a_t(u) + \lambda_na_t(u)a_t(v)t\}\lambda_nt - \bar s(u)\bar s(v)n\lambda_n\bar a_tt}{\bar s_T\{\bar s_T + n\lambda_n\bar a_tt\}}\\
& = \left\{\frac{\bar s(u)a_t(v) + \bar s(v)a_t(u) + \lambda_na_t(u)a_t(v)t - r_{uv}(\bar s)n\bar a_t}{\bar s_T + n\lambda_n\bar a_tt}\right\}\lambda_nt
\end{align*}
Using parts (2)-(4) of Lemma \ref{expv-lemma}, for sufficiently small $t$ we have
\begin{align*}
\big|r_{uv}(\vS) - r_{uv}(\bar \vs)\big| & \leq\frac{2\lambda_nm_+^3 + \lambda_nt + \rho_n(m_+^6/m_-^2)n}{n\lambda_nm_-^2 - n\lambda_nt}\lambda_nt = \frac{2m_+^3 + t + (m_+^6/m_-^2)}{m_-^2 - t}\rho_nt
\end{align*}
Therefore,
\begin{equation}
\label{ruvs-ineq}
|r_{uv}(\vS) - r_{uv}(\bar \vs)| = \rho_nO(t)
\end{equation}
as $t\rightarrow0$. By a similar argument, $|r_{uv}(\vD) - r_{uv}(\bar\vd)|= \rho_nO(t)$. It follows that 
\begin{equation}
\label{fuvsd-ineq}
|f_{uv}(\vD,\vS) - f_{uv}(\bar \vd,\bar \vs)| = \rho_nO(t).
\end{equation}
Therefore, using Equations \ref{ruvs-ineq}-\ref{fuvsd-ineq} and part (7) of Lemma \ref{expv-lemma},
\begin{align*}
V_{uv}(\vD,\vS) & := (W_{uv} - f_{uv}(\vD, \vS))^2\\
& = (W_{uv} - f_{uv}(\bar \vd, \bar \vs))^2 + 2(W_{uv} - f_{uv}(\bar \vd, \bar \vs))(f_{uv}(\bar \vd, \bar \vs) - f_{uv}(\vD, \vS))\\
& + (f_{uv}(\bar \vd, \bar \vs) - f_{uv}(\vD, \vS))^2 \\
& = V_{uv}(\bar \vd, \bar \vs)^2 + 2V_{uv}(\bar \vd, \bar \vs)(f_{uv}(\bar \vd, \bar \vs) - f_{uv}(\vD, \vS)) + (f_{uv}(\bar \vd, \bar \vs) - f_{uv}(\vD, \vS))^2 \\
& \leq V_{uv}(\bar \vd, \bar \vs)^2 + \rho_nO(t) + \rho_n^2O(t^2) = V_{uv}(\bar \vd, \bar \vs)^2 + \rho_nO(t)
\end{align*}
Define the following:
\[
V_T:=\sum_{u,v:A_{uv} = 1}V_{uv}(\vD,\vS),\;\;\;\;\;\bar V_T := \sum_{u,v:A_{uv} = 1}V_{uv}(\bar{\vd},\bar{\vs}).
\]
Since $D_T := \sum_{u\in N}D(u) = \sum_{u,v:A_{uv} = 1}1$, the above inequality implies that $V_T = \bar V_T + D_T\rho_nO(t)$. Define similarly:
\[
g_T := \sum_{u,v:A_{uv} = 1}f_{uv}(\vD, \vS)^2,\;\;\;\;\;\bar g_T := \sum_{u,v:A_{uv} = 1}f_{uv}(\bar\vd, \bar\vs)^2.
\]
Similar logic gives $g_T = \bar g_T + D_T\rho_nO(t)$. Finally, define $\bar v_T := \sum_{u,v:A_{uv} = 1}v_{uv}(\bar\vs,\bar\vd)$. Then
\begin{align*}
\big| \hat\kappa(\vD, \vS) - \kappa_\ast(\bar\vd,\bar\vs) \big| & = \left|\frac{V_T}{g_T} - \frac{\bar v_T}{\bar g_T}\right| = \left|\frac{\bar V_T + D_T\rho_nO(t)}{\bar g_T + D_T\rho_nO(t)} - \frac{\bar v_T}{\bar g_T}\right|\\
& = \left|\frac{\bar V_T + D_T\rho_nO(t) - \frac{\bar v_T}{\bar g_T}\left\{\bar g_T + D_T\rho_nO(t)\right\}}{\bar g_T + D_T\rho_nO(t)}\right|\\
& \leq \left|\frac{\bar V_T - \bar v_T}{\bar g_T + D_T\rho_nO(t)}\right| + \left|\frac{D_T\rho_nO(t) - \frac{\bar v_T}{\bar g_T}D_T\rho_nO(t)}{\bar g_T + D_T\rho_nO(t)}\right|
\end{align*}
Note that $\bar v_T/D_T$ and $\bar g_T/D_T$ are, each, by parts (5) and (6) of Lemma \ref{expv-lemma}, bounded above and below by constants independent of $A$, $t$, and $n$. Therefore, dividing through by $D_T$,
\[
\left|\frac{D_T\rho_nO(t) - \frac{\bar v_T}{\bar g_T}D_T\rho_nO(t)}{\bar g_T + D_T\rho_nO(t)}\right| \leq \frac{\rho_nO(t)}{\bar g_T/D_T + \rho_nO(t)} = \rho_nO(t)
\]
This proves part 1. For part 2, first recall that $\mu(u,B|\Theta)\equiv\mu(u,B|\vS):=\sum_{v\in B}r_{uv}(\vS)$. Therefore by Equation \ref{ruvs-ineq}, we have
\begin{equation}
\label{mu-ineq}
|\mu(u,B|\Theta) - \mu(u,B| \theta_\ast)| = \big|\sum_{v\in B} r_{uv}(\vS) - r_{uv}(\bar \vs)\big| = |B|\rho_nO(t)
\end{equation}
Recall further that
\[
\sigma(u,B| \Theta)^2 := \sum_{v\in B}r_{uv}(\vS)f_{uv}(\vD,\vS)\left(1 - r_{uv}(\vD) + \hat\kappa(\vD,\vS)\right)
\]
Using some straightforward algebra and applying Equations \ref{ruvs-ineq}-\ref{fuvsd-ineq}, we have
\begin{align}
\left|\sigma(u,B|\Theta)^2 - \sigma(u,B|\theta_\ast)^2\right| & = |B|\left(1 + \big|\hat\kappa(\vD,\vS) - \kappa_\ast(\bar\vd,\bar\vs)\big|\right)\rho_nO(t)\nonumber\\
& = |B|\rho_nO(t)\label{sig2-ineq}
\end{align}
where the second line follows from the assumption that $|\hat\kappa(\vD,\vS) - \kappa_\ast(\bar\vd,\bar\vs)|\leq \eps$. We will now bound $\sigma(u,B|\Theta)$ close to $\sigma(u,B|\theta_\ast)$ using Equation \ref{sig2-ineq} and a Taylor expansion. Define the function $h(x,\sigma) := \sqrt{\sigma^2 + x}$. For fixed $\sigma$, a Taylor expansion around $x = 0$ gives $h(x,\sigma) = \sigma + \sum_{k = 1}^\infty (-1)^k\frac{x^k}{k!\sigma^{2k-1}}$. 
Setting $x = \sigma(u,B|\Theta)^2 - \sigma(u,B|\theta_\ast)^2$ and $\sigma = \sigma(u,B|\theta_\ast)$ and applying Equation \ref{sig2-ineq}, we obtain 
\begin{align}
\sigma(u,B|\Theta) & = h(x, \sigma(u,B|\theta_\ast))\nonumber \\
& = \sigma(u,B|\theta_\ast) + \sum_{k = 1}^\infty(-1)^k\frac{|B|^k\rho_n^kO(t^k)}{k!\sigma(u,B|\theta_\ast)^{2k-1}}\label{sig-ineq0}
\end{align}
Part (8) of Lemma \ref{expv-lemma} implies that $\sigma(u,B|\theta_\ast) \asymp \sqrt{|B|\rho_n}$. Equation \ref{sig-ineq0} therefore gives
\begin{equation}
\label{sig-ineq}
\sigma(u,B|\Theta) = \sigma(u,B|\theta_\ast) + \sqrt{|B|\rho_n}O(t)
\end{equation}
using Equations \ref{mu-ineq} and \ref{sig-ineq}, we write
\begin{equation}\label{mu-sig-approx1}
\left|\frac{\mu(u,B|\Theta)}{\sigma(u,B|\Theta)} - \frac{\mu(u,B|\theta_\ast)}{\sigma(u,B|\theta_\ast)}\right| = \left|\frac{\mu(u,B|\theta_\ast) + |B|\rho_nO(t)}{\sigma(u,B|\theta_\ast) + \sqrt{|B|\rho_n}O(t)} - \frac{\mu(u,B|\theta_\ast)}{\sigma(u,B|\theta_\ast)}\right|
\end{equation}
As shorthands, define $\bar\mu_n := \mu(u,B|\theta_\ast)/|B|\rho_n$ and $\bar\sigma_n:=\sigma(u,B|\theta_\ast)/\sqrt{|B|\rho_n}$. Part (8) of Lemma \ref{expv-lemma} implies that $\bar \mu_n,\bar \sigma_n\asymp 1$. Thus, using Equation \ref{mu-sig-approx1} and dividing through by the appropriate factors,
\begin{align*}
\left|\frac{\mu(u,B|\Theta)}{\sigma(u,B|\Theta)} - \frac{\mu(u,B|\theta_\ast)}{\sigma(u,B|\theta_\ast)}\right| & = \sqrt{|B|\rho_n}\left|\frac{\bar \mu_n + O(t)}{\bar \sigma_n + O(t)} - \frac{\bar \mu_n}{\bar \sigma_n}\right|\\
& = \sqrt{|B|\rho_n}O(t)
\end{align*}
This completes part 2.
\end{proof}

The proof of Lemma \ref{initial-to-community} from the main text (below) makes use of Lemma \ref{gross-lemma} by showing that its assumption holds with high probability, for appropriate $t$.

%% Main Theorem Proof

\subsection{Proof of Theorem \ref{initial-to-community}}
Throughout, we will sometimes suppress dependence on $n$ for notational convenience. Recall that $A(u,B,\Net) := S(u,B,\Net) - \mu(u,B|\vS)$, the deviation of the CCME test statistic from its expected value under the continuous configuration model. Recalling that $\Theta := (\vD,\vS,\hat\kappa(\vD,\vS))$, define also the random $Z$-statistic
\begin{equation}\label{random-zstat}
Z(u,B,\Net|\Theta):= \frac{A(u,B,\Net)}{\sigma(u,B|\Theta)}.
\end{equation} Define the random p-value
\begin{equation}
P(u,B,\Net|\Theta) := 1 - \Phi(Z(u,B,\Net|\Theta)).
\end{equation}
The random variable $P(u,B,\Net|\Theta)$ is the random version of the p-value $p(u,B_n|\theta)$ obtained from the approximation in Equation \eqref{eq:pvalue_approx}. As a consequence of the Benjamini-Hochberg procedure, the event $\{U_\alpha(B_n,\Net) = C_n\}$ will occur if
\begin{align}
P(u,B_n,\Net_n|\Theta) & \leq q\alpha,\;\text{ for all }u\in C_n,\;\text{ and}\nonumber\\
P(u,B_n,\Net_n|\Theta) & > q\alpha,\;\text{ for all }u\notin C_n,\label{red-sea}
\end{align}
since by assumption $|C_n|>qn$. Let $h$ be the density function of a standard-Normal. By a well-known inequality for the CDF of a standard-Normal, if $Z(u,B_n,\Net_n|\Theta)>0$,
\begin{equation}\label{normal-conc-ineq}
P(u,B_n,\Net_n|\Theta)\leq \frac{1}{Z(u,B_n,\Net_n|\Theta)}h(Z(u,B_n,\Net_n|\Theta)).
\end{equation}
By symmetry, if $Z(u,B_n,\Net_n|\Theta)<0$, then
\begin{equation}\label{normal-conc-ineq2}
P(u,B_n,\Net_n|\Theta)\geq 1 + \frac{1}{Z(u,B_n,\Net_n|\Theta)}h(Z(u,B_n,\Net_n|\Theta)).
\end{equation}
We therefore analyze the concentration properties of $Z(u,B_n,\Net_n|\Theta)$ and apply Inequalities \ref{normal-conc-ineq} and \ref{normal-conc-ineq2} to show that for sufficiently large $n$, the event in Equation \ref{red-sea} occurs with high probability. We will focus on the first line of \ref{red-sea} first; the second is shown similarly. Recall that $\theta_\ast$ is the empirical population null parameters of $\Net_n$, defined after Definition \ref{pop-var}. For the derivation below we use the following shorthands: $Y\equiv S(u,B_n,\Net_n)$, $\mu \equiv \mu(u,B_n|\vS_n)$, $\sigma := \sigma(u,B_n|\Theta)$, $\bar y \equiv \expv Y$, $\bar \mu \equiv \mu(u,B_n|\theta_\ast)$, and $\bar \sigma := \sigma(u,B_n|\theta_\ast)$. Note
\begin{align}
Z(u,B_n,\Net_n|\Theta) := \frac{Y - \mu}{\sigma} = \frac{Y - \bar{\mu}}{\bar{\sigma}} - \left(\frac{\mu}{\sigma} - \frac{ \bar \mu }{ \bar \sigma }\right) & = \frac{ \bar y - \bar \mu}{ \bar \sigma} + \frac{Y - \bar y}{\bar \sigma} - \left(\frac{\mu}{\sigma} - \frac{ \bar \mu }{ \bar \sigma }\right)\nonumber\\
& \geq \frac{ \bar y - \bar \mu}{ \bar \sigma} - \left|\frac{Y - \bar y}{\bar \sigma}\right| - \left|\frac{\mu}{\sigma} - \frac{ \bar \mu }{ \bar \sigma }\right|\label{pivotal-inequality}
\end{align}
% Messier version of the derivation above.
\begin{comment}
\begin{align*}
\frac{S(u:B_n) - \mfun{u}{B_n}{\vS}}{\sfun{u}{B_n}{\vS}{\vD}} & = \frac{S(u:B_n) - \mfun{u}{B_n}{\tilde\vs}}{\tsfun{u}{B_n}{\tilde\vs}{\tilde\vd}} - \left(\frac{\mfun{u}{B_n}{\vS}}{\sfun{u}{B_n}{\vS}{\vD}} - \frac{\mfun{u}{B_n}{\tilde\vs}}{\tsfun{u}{B_n}{\tilde\vs}{\tilde\vd}}\right)\\
& = \frac{\expv S(u:B_n) - \mfun{u}{B_n}{\tilde\vs}}{\tsfun{u}{B_n}{\tilde\vs}{\tilde\vd}} + \frac{S(u:B_n) - \expv S(u:B_n)}{\tsfun{u}{B_n}{\tilde\vs}{\tilde\vd}}\\
& - \left(\frac{\mfun{u}{B_n}{\vS}}{\sfun{u}{B_n}{\vS}{\vD}} - \frac{\mfun{u}{B_n}{\tilde\vs}}{\tsfun{u}{B_n}{\tilde\vs}{\tilde\vd}}\right)
\end{align*}
Therefore,
\begin{align}
\frac{S(u:B_n) - \mfun{u}{B_n}{\vS}}{\sfun{u}{B_n}{\vS}{\vD}} & \geq \frac{\expv S(u:B_n) - \mfun{u}{B_n}{\tilde\vs}}{\tsfun{u}{B_n}{\tilde\vs}{\tilde\vd}} - \left|\frac{S(u:B_n) - \expv S(u:B_n)}{\tsfun{u}{B_n}{\tilde\vs}{\tilde\vd}}\right|\nonumber\\
& - \left|\frac{\mfun{u}{B_n}{\vS}}{\sfun{u}{B_n}{\vS}{\vD}} - \frac{\mfun{u}{B_n}{\tilde\vs}}{\tsfun{u}{B_n}{\tilde\vs}{\tilde\vd}}\right|\label{ineq1}
\end{align}
\end{comment}
Define
\[
\bar z(u,B_n|\theta_\ast):=\frac{\bar y - \bar \mu}{\bar \sigma} = \lambda_n\frac{\tilde a(u,B_n|\bar \vs)}{\sigma(u,B_n|\theta_\ast)}
\]
where $\tilde a(u,B_n|\bar \vs)$ is the normalized population version of $A(u,B_n|\vS)$, as defined in Equation \ref{dev-def} from the main text. The definition above works with Equation \ref{pivotal-inequality} to produce the illustrative inequality
\begin{equation}\label{pivotal-inequality2}
Z(u,B_n,\Net_n|\Theta) \geq \bar z(u,B_n|\theta_\ast) - \left|\frac{Y - \bar y}{\bar \sigma}\right| - \left|\frac{\mu}{\sigma} - \frac{ \bar \mu }{ \bar \sigma }\right|.
\end{equation}
Inequality \ref{pivotal-inequality2} exemplifies that, if the right-hand terms vanish, $Z(u,B_n,\Net_n|\Theta)$ can be approximated by a population version. Our analysis therefore reduces to bounding the right-hand order terms in probability.  

Explicitly, consider that by part (8) of Lemma \ref{expv-lemma}, there exists $m_2>0$ such that $\sigma(u,B_n|\theta_\ast)^2\leq m_2 n\rho_n = m_2\lambda_n$. Combining this with the crucial assumption on $\tilde a(u,B_n)$ from line \ref{favors} from the main text, we have that for all $u\in C_n$, 
\begin{equation}\label{eq:population-z-rate}
\bar z(u,B_n|\theta_\ast)= \lambda_n\frac{\tilde a(u,B_n|\bar \vs)}{\sigma(u,B_n|\theta_\ast)} \geq \sqrt{\lambda_n}\frac{\Delta}{\sqrt m_2}
\end{equation}
Therefore, the rest of the proof is mainly dedicated to showing that the final two terms in line \eqref{pivotal-inequality2} are $o_P(\sqrt{\lambda_n})$. This will imply that $Z(u,B_n,\Net_n|\Theta) = \Omega_P(\sqrt{\lambda_n})$ and, using Inequality \ref{normal-conc-ineq}, that $\{P(u,B_n,\Net_n|\Theta)\leq q\alpha,\;\forall\;u\in C_n\}$ has probability approaching 1.\\

%%% Step 1
\noindent \emph{Step 1: $|\frac{\mu}{\sigma} - \frac{ \bar \mu }{ \bar \sigma }| = O_P(\sqrt{\log n})$}\\

For $t>0$, define the event
\begin{equation}\label{event1}
\cE_1(t) := \left\{\underset{u\in N}{\max}\left|S(u) - \bar s(u)\right|,\underset{u\in N}{\max}\left|D(u) - \bar d(u)\right|\leq \lambda_nt\right\}
\end{equation}
Fix arbitrary $b > 0$ independent of all other quantities and define $t_n(b):=\sqrt{\frac{\taun}{\lambda_n}}$. Note that $t_n(b)\rightarrow0$ for any $b$, by the assumptions of the Theorem. Recall that $D_T:=\sum_{u\in N}D(u)$, the (random) total degree. For notational convenience, let $E := \{\text{pairs } u,v:A_{uv} = 1\}$. By part 1 of Lemma \ref{gross-lemma}, the event $\cE_1(t_n(b))$ implies
\begin{equation}\label{lemma3-recall1}
\left|\hat\kappa(\vD,\vS) - \kappa_\ast(\bar\vd,\bar\vs)\right| = \left|\frac{\sum_{E} V_{uv}(\bar \vd,\bar \vs) - v_{uv}(\bar\vd,\bar\vs)}{\sum_{E}f_{uv}(\bar\vd,\bar\vs)^2 + D_T\rho_nO(t_n(b))}\right| + \rho_nO(t_n(b))
\end{equation}
By Lemma \ref{expv-lemma} part (5),
\[
0\leq V_{uv}(\bar \vd, \bar \vs)\leq (\eta m_+^2/m_-^2 + m_+^{10}/m_-^{3})^2.
\]
Recall that $v_{uv}(\bar \vd, \bar \vs) := \expv V_{uv}(\bar \vd, \bar \vs)$, and that the edge weights that comprise the (upper-triangle of the) weight matrix $W$ are independent. For a fixed adjacency matrix $A$, Bernstein's Inequality therefore gives
\begin{equation}\label{bern0}
\mathbb{P}\left(\left|\sum_E V_{uv}(\bar \vd, \bar \vs) -  v_{uv}(\bar\vd,\bar\vs)\right| > \sqrt{\taun}\;\Bigg|\; A\right)\leq 2\exp\left\{\frac{-2\taun}{2a_1 + \frac{2}{3}a_2\sqrt{\taun }}\right\}
\end{equation}
Now by Lemma \ref{expv-lemma} part (6), $\sum_E\ffun{\bar\vd}{\bar\vs}^2\geq D_T\frac{m_-^{10}}{m_+^3}$. Thus
\[
\sum_E\ffun{\bar\vd}{\bar\vs}^2 + D_T\rho_nO(t_n(b))\geq D_T\frac{m_-^{10}}{m_+^3}/2
\]
for large enough $n$, since $\rho_nt_n(b)\rightarrow 0$. Therefore there exist constants $a_1,a_2>0$ depending only on $m_+$, $m_-$, and $\eta$ such that
\begin{equation}\label{bern1}
\mathbb{P}\left(\left|\frac{\sum_E V_{uv}(\bar \vd, \bar \vs) -  v_{uv}(\bar\vd,\bar\vs)}{\sum_E\ffun{\bar\vs}{\bar\vd}^2 + D_T\rho_nO(t_n(b))}\right| > \sqrt{\tfrac{\taun}{D_T}}\;\Bigg|\; A\right)\leq 2\exp\left\{\frac{-2\taun}{2a_1 + \frac{2}{3}a_2\sqrt{\frac{\taun}{D_T}}}\right\}
\end{equation}
The above expression is conditional on a fixed adjacency matrix $A$. We now bound in probability the functionals of $A$ on which the expression depends. It is easily derivable from the statement of the WSBM and Assumption \ref{bounded-parameter-assumption} that there exist constants $a_3, a_4$ depending on $m_+$ and $m_-$ such that $\expv(D_T) = a_3n\lambda_n$ and $\var(D_T) = a_4n\lambda_n$. Therefore, by another application of Bernstein's Inequality,
\begin{equation}\label{bern2}
\mathbb{P}\left(\left|D_T - a_3n\lambda_n\right| > \sqrt{n\lambda_n\taun}\right)\leq 2\exp\left\{\frac{-2\taun}{2a_4 + \frac{2}{3}\sqrt{\frac{\taun}{n\lambda_n}}}\right\}
\end{equation}
Applying this to inequality \eqref{bern1}, the law of total probability gives
\begin{align}
&\mathbb{P}\left(\left|\frac{\esum V_{uv}(\bar \vd, \bar \vs) -  v_{uv}(\bar\vd,\bar\vs)}{\esum\ffun{\bar\vs}{\bar\vd}^2 + D_T\rho_nO(t_n(b))}\right| > \sqrt{\frac{\taun}{a_3n\lambda_n - \sqrt{n\lambda_n\taun}}}\right)\nonumber\\
& \leq 2\exp\left\{\frac{-2\taun}{2a_1 + \frac{2}{3}a_2\sqrt{\frac{\taun}{a_3n\lambda_n - \sqrt{n\lambda_n\taun}}}}\right\} + 2\exp\left\{\frac{-2\taun}{2a_4 + \frac{2}{3}\sqrt{\frac{\taun}{n\lambda_n}}}\right\} = O(n^{-b})\label{theta-conc-ineq}
\end{align}
for sufficiently large $n$. Along with Equation \eqref{lemma3-recall1}, this implies there exists a constant $A_0$ depending on parameter constraints such that
\begin{equation}\label{final-theta-ineq1}
\mathbb{P}\left\{\left|\hat\kappa(\vD,\vS) - \kappa_\ast(\bar\vd,\bar\vs)\right| \leq A_0\left(\sqrt{\frac{b\log n}{n\lambda_n}} + \rho_nt_n(b)\right)\right\} \geq \mathbb{P}(\cE_1(t_n(b))) - O(n^{-b})
\end{equation}
for sufficiently large $n$. We now assess $\mathbb{P}(\cE_1(t_n(b)))$. Note that for all $u\in N$, $\text{Var}(S(u)) = O(\lambda_n)$. Furthermore, recall from Inequality \ref{Wuv-ineq} (in the proof of Lemma \ref{expv-lemma}) that $W_{uv}\leq m_+^2\eta/m_-^2$ for all $u,v\in N$. For fixed $b>0$, Bernstein's Inequality therefore gives, for any $u\in N$,
\begin{equation}\label{bern0S}
\mathbb{P}\left(\left|S(u) - \bar s(u)\right|>\sqrt{\taun\lambda_n}\right)\leq 2\exp\left\{\frac{-2a_0\taun}{2 + \frac{2}{3}\sqrt{\frac{\taun}{\lambda_n}}}\right\},
\end{equation}
where $a_0$ is a constant independent of $n$. The constant $a_0$ may be chosen so that, similarly,
\begin{equation}\label{bern0D}
\mathbb{P}\left(\left|D(u) - \bar d(u)\right|>\sqrt{\taun\lambda_n}\right)\leq 2\exp\left\{\frac{-2a_0\taun}{2 + \frac{2}{3}\sqrt{\frac{\taun}{\lambda_n}}}\right\}
\end{equation}
Applying a union bound, equations \eqref{bern0S} and \eqref{bern0D} give
\begin{align}
\mathbb{P}(\cE_1(t_n(b))) & \geq 1 - 2n\exp\left\{\frac{-2a_0\taun}{2 + \frac{2}{3}\sqrt{\frac{\taun}{\lambda_n}}}\right\} - 2n\exp\left\{\frac{-2a_0\taun}{2 + \frac{2}{3}\sqrt{\frac{\taun}{\lambda_n}}}\right\}\nonumber\\
& = 1 - O(n^{-b + 1})\label{bern0}
\end{align}
for sufficiently large $n$. Returning to the inequality in \eqref{final-theta-ineq1}, we therefore have
\begin{align}
\mathbb{P}\left\{\left|\hat\kappa(\vD,\vS) - \kappa_\ast(\bar\vd,\bar\vs)\right| \leq A_0\left(\sqrt{\tfrac{b\log n}{n\lambda_n}} + \rho_nt_n(b)\right)\right\} & \geq \mathbb{P}(\cE_1(t_n(b))) - O(n^{-b})\nonumber\\
& \geq 1 - O(n^{-b + 1})\label{theta-subfinal-bern}
\end{align}
for sufficiently large $n$. Recall that by assumption, $\lambda_n/\log n\rightarrow\infty$. Thus $t_n(b)\rightarrow0$, and
\[
\sqrt{\tfrac{b\log n}{n\lambda_n}} + \rho_nt_n(b) = t_n(b)/\sqrt{n} + \rho_nt_n(b)\leq 1/\sqrt{n} = o(1).
\]
Thus, Inequality \ref{theta-subfinal-bern} implies that 
\begin{equation}\label{theta-final-bern}
\mathbb{P}\left(\left|\hat\kappa(\vD,\vS) - \kappa_\ast(\bar\vd,\bar\vs)\right|\leq \eps\right)\geq 1 - O(n^{-b + 1}),
\end{equation}
for sufficiently large $n$. For $\eps>0$, define the event $\cE_2(\eps) := \left\{\left|\hat\kappa(\vD,\vS) - \kappa(\bar\vd,\bar\vs)\right|\leq \eps\right\}$. By part 2 of Lemma 3, the event $\cE_1(t_n(b))\cap\cE_2(\eps)$ implies
\begin{align}
\left|\frac{\mu}{\sigma} - \frac{\bar{\mu}}{\bar \sigma}\right|:=\left|\frac{\mfun{u}{B_n}{\vS}}{\sfun{u}{B_n}{\Theta}} - \frac{\mfun{u}{B_n}{\bar\vs}}{\sfun{u}{B_n}{\theta_\ast}}\right| & = \sqrt{|B_n|\rho_n}O(t_n(b))\nonumber\\
& \leq \sqrt{\lambda_n} O(t_n(b))\nonumber.\\
& = O(\sqrt{\taun})\label{rbound1}
\end{align}
Therefore, there exists a constant $A_2>0$ such that, by Inequalities \ref{bern0} and \ref{theta-final-bern},
\begin{align}
\mathbb{P}\left(\left|\frac{\mu}{\sigma} - \frac{\bar{\mu}}{\bar \sigma}\right| \leq A_2\sqrt{\taun}\right) & = 1 - O(n^{-b + 1})\label{third-term-bound}
\end{align}
for sufficiently large $n$. This completes Step 1.\\

%%% Step 2
\noindent \emph{Step 2: $|\frac{Y - \bar y}{\bar \sigma}| = O_P(\sqrt{\log n})$.}\\

Note that, as for Inequality \ref{bern0S}, Bernstein's Inequality gives
\begin{equation}\label{bernSuC}
\mathbb{P}\left(\left|S(u,B_n,\Net_n) - \expv S(u,B_n,\Net_n)\right|>\sqrt{\taun\lambda_n}\right)\leq 2\exp\left\{\frac{-2a_0\taun}{2 + \frac{2}{3}\sqrt{\frac{\taun}{\lambda_n}}}\right\}
\end{equation}
By Lemma \ref{expv-lemma} part (8), there exists $m_2>0$ such that $\sfun{u}{B_n}{\theta_\ast}^2\leq m_2\lambda_n$. Thus,
\[
\left|\frac{ Y - \bar y }{\bar \sigma}\right| := \left|\frac{S(u,B_n,\Net_n) - \expv S(u,B_n,\Net_n)}{\sfun{u}{B_n}{\theta_\ast}}\right| \geq \left|\frac{S(u,B_n,\Net_n) - \expv S(u,B_n,\Net_n)}{m_2\sqrt{\lambda_n}}\right|,
\]
so by Inequality \ref{bernSuC}, we have for sufficiently large $n$ that
\begin{equation}\label{rbound2}
\mathbb{P}\left(\left|\frac{ Y - \bar y }{\bar \sigma}\right|\leq \sqrt{\frac{b\log n}{m_2}}\right)\geq 1 - O(n^{-b}).
\end{equation}
This completes Step 2.\\

We now recall inequality \ref{pivotal-inequality2}:
\[
Z(u,B_n,\Net_n|\Theta) \geq \bar z(u,B_n|\theta_\ast) - \left|\frac{Y - \bar y}{\bar \sigma}\right| - \left|\frac{\mu}{\sigma} - \frac{ \bar \mu }{ \bar \sigma }\right|.
\]
In step 1, we showed that there exists a constant $A_2$ depending only on the fixed WSBM model parameters such that for any fixed $b>1$, for large enough $n$, $\left|\frac{\mu}{\sigma} - \frac{ \bar \mu }{ \bar \sigma }\right|\leq A_2\sqrt{b\log n}$ with probability $1 - O(n^{-b + 1})$. In step 2, we showed that there exists a constant $m_2$ depending only on the fixed WSBM model parameters such that for any fixed $b>1$, for large enough $n$, $|\frac{Y - \bar y}{\bar \sigma}|\leq \sqrt{b\log n/m_2}$ with probability $1 - O(n^{-b})$. Recall furthermore from inequality \ref{eq:population-z-rate} that $\bar z(u, B_n|\theta_\ast) \geq \Delta\sqrt{\lambda_n/m_2}$, where $\Delta$ is from condition \ref{favors} in the statement of the Theorem. We can therefore write that for any fixed $b>1$, for large enough $n$,
\[
Z(u,B_n,\Net_n|\Theta) \geq \Delta\sqrt{\lambda_n/m_2} - \sqrt{b\log n/m_2} - A_2\sqrt{b\log n} = A_3\sqrt{\lambda_n} - A_4\sqrt{b\log n}
\]
with probability at least $1 - O(n^{-b + 1})$. Now, by assumption, $|C_n|\geq qn$. Therefore, using Inequality \ref{normal-conc-ineq} and a union bound, we can write that for any fixed $b>1$, for large enough $n$,
\begin{equation}\label{inC-pval-ineq}
\underset{u\in C_n}{\max}\;P(u,B_n,\Net_n|\Theta)\leq \exp\{-(A_3\sqrt{\lambda_n} - A_4\sqrt{b\log n})^2\}
\end{equation}
with probability at least $1 - O(n^{-b + 2})$. Note that for any fixed $b$, the right-hand-side of inequality \ref{inC-pval-ineq} vanishes, due to the assumption that $\lambda_n/\log n\rightarrow\infty$. Thus, for $b>2$, inequality \ref{inC-pval-ineq} implies that for large enough $n$ (now depending on choice of $b$), the event $\{P(u,B_n,\Net_n|\Theta)\leq q\alpha,\;\forall\;u\in C_n\}$ has probability $1 - O(n^{-b + 2})\rightarrow1$.

It can be similarly shown that the second half of the event in \eqref{red-sea} has probability approaching 1. Instead of Inequality \ref{pivotal-inequality2} we (similarly) derive
\begin{equation}
\label{pivotal-inequality2b}
Z(u,B_n,\Net_n|\Theta) \leq \bar z (u,B_n|\theta_\ast) + \left|\frac{Y - \bar y}{\bar \sigma}\right| + \left|\frac{\mu}{\sigma} - \frac{ \bar \mu }{ \bar \sigma }\right|
\end{equation}
This is useful because if $u\notin C_n$, assumption \eqref{favors} ensures that $\devnt{u}{B_n}{\bar\vs}<-\Delta$, and hence 
\[
\bar z(u,B_n|\theta_\ast):=\frac{\bar y - \bar \mu}{\bar \sigma} = \lambda_n\frac{\tilde a(u,B_n|\bar \vs)}{\sigma(u,B_n|\theta_\ast)} \leq \lambda_n\frac{-\Delta}{\sigma(u,B_n|\theta_\ast)} \leq \sqrt{\lambda_n}\frac{-\Delta}{\sqrt m_1}
\]
where the last inequality follows from part (8) of Lemma \ref{expv-lemma}. Steps 1 and 2 therefore work to show that for any fixed $b>1$, for large enough $n$,
\[
Z(u,B_n,\Net_n|\Theta) \leq -\Delta\sqrt{\lambda_n/m_2} + \sqrt{b\log n/m_2} + A_2\sqrt{b\log n} = A_3\sqrt{\lambda_n} - A_4\sqrt{b\log n}
\]
With probability $1 - O(n^{-b + 1})$. Inequality \ref{normal-conc-ineq2} then implies that
\begin{equation}\label{inC-pval-ineq2}
\mathbb{P}\left(\underset{u\notin C_n}{\max}\;P(u,B_n,\Net_n|\Theta)\geq 1-\exp\{-(A_3\sqrt{\lambda_n} - A_4\sqrt{b\log n})^2\}\right)\geq 1 - O(n^{-b + 2})
\end{equation}
With reasoning identical to the result for $u\in C_n$, this implies that for any $b>2$, for large enough $n(b)$, the event $\{P(u,B_n,\Net_n|\Theta)>q\alpha,\;\forall\;u\notin C_n\}$ has probability at least $1 - O(n^{-b + 2})\rightarrow1$. Applying a union bound to the event in \eqref{red-sea} completes the proof.\qed

\subsection{Proof of Theorem \ref{thm:consistency}}

	We will show that if the condition in \eqref{M-favors} holds, then the condition in \eqref{favors} from Theorem \ref{initial-to-community} holds when $B_n = C_n = C_{j,n}$ simultaneously across all $j\in\{1,2,\ldots,K\}$. This involves representing \eqref{favors} in terms of the model parameters when $B_n = C_n = C_{j,n}$. Specifically, we derive the normalized population deviation $\tilde a(u, C_{j,n}|\bar \vs):= (\expv S(u,C_{j,n},\Net_n) - \mu(u,C_{j,n}|\bar \vs))/\lambda_n$. First, note that for any fixed $j\leq K$, part (1) of Lemma \ref{expv-lemma} gives
	\[
	\sum_{v\in C_{j,n}}\bar s(v)\; = \;\lambda_n\bra\tilde \vpi,\mH_j\ket\cdot\sum_{v\in C_{j,n}}\psi(u)\; = \;n\lambda_n\bra\tilde \vpi,\mH_j\ket \tilde\pi_j
	\]
	and thus
	\[
	\bar s_T := \sum_{v\in N}\bar s(v) = \sum_{j=1}^K\sum_{v\in C_{j,n}}\bar s(v) = n\lambda_n\sum_{j=1}^K\bra\tilde \vpi,\mH_j\ket\tilde \pi_j = n\lambda_n\tilde \vpi^t \mH \vpi.
	\]
	Therefore, again applying part (1) of Lemma \ref{expv-lemma},
	\begin{align*}
	\mu(u,C_{j,n}|\bar \vs) := \sum_{v\in C_{j,n}} r_{uv}(\bar \vs) = \bar s(u) \sum_{v\in C_{j,n}}\frac{\bar s(v)}{\bar s_T} & = \bar s(u)\frac{\bra\tilde \vpi,\mH_j\ket\tilde \pi_j}{\tilde \vpi^t \mH \tilde \vpi}\\
	& = \lambda_n\psi(u)\frac{\bra\tilde \vpi,\mH_{c(u)}\ket\bra\tilde \vpi,\mH_j\ket\tilde \pi_j}{\tilde \vpi^t \mH \tilde \vpi}.
	\end{align*}
	Secondly,
	\begin{align*}
	\expv S(u,C_{j,n},\Net_n) & = \sum_{v\in C_{j,n}}\expv W_{uv}= \sum_{v\in C_{j,n}}\rho_nr_{uv}(\psi)\mH_{c(u)j} = \lambda_n\psi(u)\mH_{c(u)j}\tilde\pi_j.
	\end{align*}
	Thus,
	\begin{align}
	\tilde a(u, C_{j,n}|\bar \vs) & :=\frac{\expv S(u,C_{j,n},\Net_n) - \mu(u,C_{j,n}|\bar \vs)}{\lambda_n} \nonumber\\
	& = \psi(u)\tilde \pi_j\left(\mH_{c(u)j} - \frac{\bra\tilde \vpi,\mH_{c(u)}\ket\bra\tilde \vpi,\mH_j\ket}{\tilde \vpi^t \mH \tilde \vpi}\right).\label{npd-condition}
	\end{align}
	If $u\in C_{i,n}$, the expression in the parentheses from the right-hand-side of \eqref{npd-condition} is the $i,j$-th element of the matrix $\mH - \mH\tilde{\mPi}\mH/\tilde \vpi^t\mH\tilde \vpi$, with $\tilde \mPi:=\tilde \vpi\tilde \vpi^t$. By Assumption \ref{bounded-parameter-assumption}, $\psi(u)\geq m_-$ for all $u\in N$ and $i\leq K$, and $\tilde{\pi}_j$ is fixed. Thus, \eqref{M-favors} ensures that \eqref{favors} holds when $C_n = C_{j,n}$, simultaneously for $j\leq K$. Assumption \ref{bounded-parameter-assumption} also ensures that there exists $q>0$ such that for all $j\leq K$ and $n>1$, $|C_{j,n}|>qn$. This allows us to apply Theorem \ref{initial-to-community} to the sequences $B_n = C_n = C_{j,n}$, for each $j\leq K$. A union bound proves the result.\qed

%% file: section_files/bw_Appendix_nontheory2.tex
\section{Cycles in Fixed Point Search} \label{app:cycle}
As remarked in Section \ref{ss:FPS}, it is possible for the SCS algorithm to reach a stable sequence $C_1,\ldots, C_J$ that is traversed by the update $U_\alpha(\cdot,\Net)$. If this happens, we apply the following routine to re-start the algorithm, or return the union of the sequence:
\begin{enumerate}[1.]
\item If $C_{i}\cap C_{i + 1} = \phi$ for any $i\leq J$, or if $C_J\cap C_1 = \phi$, terminate the iterations and do not extract a community.
\item{Otherwise, define $C^\ast = \cup_{i = 1}^J C_{i}$, and:
  \begin{enumerate}[(a)]
  \item If $C^\ast$ has been visited previously by SCS, extract $C^\ast$ into $\cC$.
  \item Otherwise, re-initialize with $C^\ast$.
  \end{enumerate}
}
\end{enumerate}

\section{Filtering of $\mathcal{B}_0$ and $\mathcal{C}$} \label{app:filtering}
To filter through $\mathcal{B}_0$ and $\mathcal{C}$, we use an inference procedure based on a set-wise $z$-statistic, analogous to the node-set $z$-statistic presented in Section \ref{theory}. Define $S(B):=\sum_{v\in B}S(v,B)$. Note that $S(B)$ has an easily derivable expectation and standard deviation under the continuous configuration model, which we denote (respectively) by $\mu(B|\theta)$ and $\sigma(B|\theta)$. We define the corresponding $z$-statistic and an approximate p-value by
\[
z(B|\theta) := \frac{S(B) - \mu(B|\theta)}{\sigma(B|\theta)},\;\;\;\;\;\;\; p(B|\theta):= 1 - \Phi(z(B|\theta))
\]
Before initializing the SCS algorithm on sets in $\cB_0$, we compute the p-value above for each member set, and remove any that are not significant at FDR level $\alpha = 0.05$. This greatly reduces the number of extractions CCME must perform, and reduces the probability of convergence on small, spurious communities.

We also use $z(B|\theta)$ to filter near-matches in $\mathcal{C}$, once all SCS extractions have terminated and empty sets removed. To do so, we require an overlap ``tolerance" parameter $\tau\in[0, 1]$. First, we create a (non-symmetric) $|\mathcal{C}|\times|\mathcal{C}|$ matrix $O$ with general element $O_{ij} := |C_i\cap C_j|/|C_i|$, which measures the proportional overlap of $C_i$ into $C_j$. After setting the diagonal of $O$ to zero, the filtering proceeds as follows:
\begin{enumerate}[1.]
	\item Find indices $i\ne j$ corresponding to the maximum entry of $O$.
	\item If $O_{ij}<\tau$, terminate filtering.
	\item Remove either $C_i$ or $C_j$ from $\mathcal{C}$, whichever has the smaller $z(B|\theta)$.
	\item Re-compute $O$, set its diagonal to zero, and return to step 1.
\end{enumerate}
For all simulations and real-data analyses in this paper, we employed this algorithm with $\tau = 0.9$.  To further decrease the computation time of CCME, as we proceed through $\mathcal{B}_0$, we skip sets that were formed from nodes that have already been extracted into $\mathcal{C}$. We find that, in practice, none of these adjustments harm CCME's ability to find statistically significant overlapping communities. Indeed, the simulation results mentioned in Section \ref{Simulations:overlapping} show that CCME outperforms competing methods with overlap capabilities.

%--------------------------------------------------------------------------------------------------------------------

%% file: supplemental1_section_files/sim_framework.tex
\section{Simulation framework}\label{sim-framework}

Here we describe the benchmarking simulation framework used in Section \ref{Simulations}. In Table \ref{tab:sim_params}, we list and name parameters controlling the network model:

\begin{table}[!htb]
\caption{\label{tab:sim_params} Simulation model parameters}
\centering
\fbox{
\scriptsize
\begin{tabular}{l|l}
$n$: Number of nodes in communities &$n_b$: Number of nodes in background\\\hline
$m_{\max}$: Max community size &$m_{\min}$: Min community size\\\hline
$\tau_1$: Power-law for degree parameters &$\tau_2$: Power-law for community sizes\\\hline
$k$: Mean of degree parameter power-law &$k_{\max}$: Maximum degree parameter\\\hline
$s_e$: Within-community edge signal &$s_w$: Within-community weight signal\\\hline
$o_n$: Number of nodes in multiple communities &$o_m$: Number of memberships for overlap nodes\\\hline
$F$: Distributions of edge weights &$\sigma^2$: Variance parameter
 for $F$\\\hline
$\beta$: Power-law for strength parameters\ &
\end{tabular}}
%\caption{Parameters controlling the generation of a random synthetic network}
\end{table} 

\subsection{Simulation of community nodes}\label{comm_nodes}

The framework is capable of simulating networks with or without background nodes. We first describe the simulation procedure without background nodes, i.e.\ with $n_b = 0$. Later, we describe how to simulate a network with background nodes, which involves a slight modification to the procedure in this subsection. Regardless of the presence of background nodes, the first step is to determine community sizes and node memberships.

\subsubsection{Community structure and node degree/strength parameters}\label{comm-sizes}
Here we describe how to obtain a cover $\mathcal{C} := \{C_1, \ldots C_K\}$ of $n$ nodes. The following steps to obtain $\mathcal{C}$ are almost exactly as those from the LFR benchmark in \cite{lancichinetti2009benchmarks}, used extensively in \cite{lanc11} and \cite{xie13}:
\begin{enumerate}
\item Each of the $o_n$ overlapping nodes will have $o_m$ memberships. Let $n_m := n + o_n(o_m - 1)$ be the number of node \emph{memberships} present in the network.
\item Draw community sizes from a power law with maximum value $m_{\max}$, minimum value $m_{\min}$, and exponent $-\tau_2$, until the sum of community sizes is greater than or equal to $n_m$. If the sum is greater than $n_m$, we reduce the sizes of the communities proportionally until the sum is equal to $n_m$. 
\item Form a bipartite graph of community markers on one side and node markers on the other. Each community marker has number of empty node slots given by step (b), and each node has a number of memberships given by step (a). Sequentially pair node memberships and community node slots uniformly at random, without replacement, until every node membership is paired with a community.
\end{enumerate}
With the community assignments in hand, simulation of the network proceeds according to the Weighted Stochastic Block Model as outlined in Section \ref{Simulations}. We describe choices for particular components of this model in the following subsection.

\subsubsection{Simulation of edges and weights}\label{comm_model}
As described in Section \ref{Simulations}, we set the $\mP$ and $\mM$ matrices to have diagonals equal to $s_e$ and $s_w$ (respectively, see Table \ref{tab:sim_params}), and off-diagonals equal to 1. We note that this homogeneity facilitates creating networks with overlapping communities. With variance in the diagonal of $\mP$, for example, it would not be obvious with what probability to connect overlapping nodes that overlap to two of the same communities, simultaneously. It remains to obtain the strength and degree propensity parameters $\psi$ and $\phi$; we do so analogously to the simulation framework in \cite{lanc11}. We first draw $\phi$ from a power law with exponent $\tau_1$, mean $k$, and maximum $k_{\max}$ (see Table \ref{tab:sim_params}). Next we set $\psi$ by the formula $\psi(u) = \phi(u)^{\beta + 1}$. 

It is worth noting here that, under the model given below, the expected degree of node $u$ is \emph{approximately} $\phi(u)$ and the expected strength \emph{approximately} $\psi(u)$. Therefore, heterogeneity/skewness in $\phi$ and $\psi$ induce heterogeneity/skewness in the degrees and strengths of the simulated networks. However, by scaling $\phi$ and $\psi$, we can force the total expected degree and total expected strength of the simulated networks to exactly match $\phi_T$ and $\psi_T$, respectively. The scaling constants depend on $\mP$ and $\mM$ and are easily derivable from the model's generative algorithm (described in Section \ref{WSBM}).

\subsubsection{Parameter settings}\label{settings}
Here we list the ``default" settings of the simulation model, mentioned in Section \ref{Simulations}. The following choices for parameters were made regardless of the simulation setting: $\tau_2 = -2$, $k = \sqrt{n}$, $k_{\max} = 3k$ (three settings which make the degree/strength distributions skewed and the network sparse), $\beta = 0.5$ (to induce a non-trivial power law between strengths and degrees), $\tau_1 = -1$, $m_{min} = n / 5$, $m_{max} = 3m_{max}/2$ (settings which produce between about 3 and 7 communities per network with skewed size distribution), and $\sigma^2 = 1/2$. Other parameter choices are specific to the simulation settings described in Section \ref{Simulations}.

\subsection{Background node simulation}\label{bg_nodes}
If $n_b > 0$, we generate a network with $n$ community nodes, and then add $n_b$ background nodes, generating all remaining edges and weights according to the continuous configuration null model introduced in the main text. First, we obtain node-wise parameters for all $n + n_b$ nodes, yielding vectors $\phi$ and $\psi$ as in subsection \ref{comm_nodes}. In a simulated network without background, $\phi(u)$ and $\psi(u)$ are approximately $\expv[{d}(u)]$ and $\expv[{s}(u)]$, respectively. To ensure that this remains the case in a network for which background nodes are added after the simulation of community nodes, we must split up each $\degp(u)$ and $\strp(u)$ into community and background portions. A few other adjustments must also be made after the simulation of community nodes. To this end, define  

\begin{itemize}
\item{ $N_C := \{1, \ldots, n\}$; $\; N_B := \{n + 1, \ldots, n+n_b\}$ (community and background node sets)
%  \begin{itemize}
%  \item[$\rightarrow$] community and background %node sets
%  \end{itemize}
}
\item{ $\degp_{C, T} := {\sum}_{{N_C}}\degp(u)$; $\; \degp_{B, T} := {\sum}_{N_B}\degp(u)$ (target total degrees of community and background nodes)
%  \begin{itemize}
%  \item[$\rightarrow$] target total degrees of %community and background nodes
%  \end{itemize}
}
\item{ $\degp_C(u) := \frac{\degp_{C, T}}{\degp_T}\degp(u)$; $\; \degp_{B}(u) := \frac{\degp_{B, T}}{\degp_T}\degp(u)$ (target edge-counts between $u$ and the community and background nodes)
%  \begin{itemize}
%  \item[$\rightarrow$] target edge-counts between %$u$ and the community and background nodes
%  \end{itemize}
}
\item{ $\degp_{1, T} := {\sum}_{{N_C}}\degp_C(u)$; $\; \degp_{2, T} := {\sum}_{N_B}\degp_B(u)$ (target total degrees of community and background \emph{subnetworks})
%  \begin{itemize}
%  \item[$\rightarrow$] target total degrees of %community and background \emph{subnetworks}
%  \end{itemize}
}
\item{ $d_C^o(u) := {\sum}_{v\in{N_C}}A_{uv}$; $\; d_B^o(u) := {\sum}_{v\in{N_B}}A_{uv}$ (observed edge-counts between $u$ and the community and background nodes)
%  \begin{itemize}
%  \item[$\rightarrow$] observed edge-counts %between $u$ and the community and background nodes
%  \end{itemize}
}
\end{itemize}
The above definitions exist analogously for the strength parameters $\strp$ (replacing ``$d$" with ``$s"$ where appropriate). The word ``target" above indicates that we will set up the background simulation model so that these values are the approximate expected values of the graph statistics they represent. 

\subsubsection{Adjusted community-node simulation model}
The only adjustment to be made to the simulation of community nodes, described in subsection \ref{comm_model}, is that the degree and strength parameters are set to a certain \emph{fraction} of their original values. This accounts for the eventual addition of background nodes, where the remaining (random) part of each nodes degree and strength is to be simulated. So, the community-node simulation (if background nodes are to be added later) follows the process described in subsection \ref{comm_nodes} with degree parameters $\{\degp_C(1), \ldots, \degp_C(n)\}$ and strength parameters $\{\strp_C(1)\ldots \strp_C(n)\}$. 

\subsubsection{Edges and weights for background}
For the simulation of the background nodes (following the community nodes) our goal is to specify adjusted degree/strength parameters $\degp'$ and $\strp'$ given the observed edge-sums $\{d_C^o(1), \ldots, d_C^o(n)\}$ and weight-sums $\{s_C^o(1), \ldots, s_C^o(n)\}$ from the community nodes. In what follows we describe this specification for $\degp'$ only; the specification for $\strp'$ is exactly analogous. We first represent $\degp_T'$, which we have yet to determine, into community and background totals:
\[\degp_T' = \degp_{C, T}' + \degp_{B, T}'\]
Since the background subnetwork has not yet been generated, we make the specification $\degp'(u) := \degp(u)$ for all $u\in N_B$, and hence $\degp_{B, T}' = \degp_{B, T}$ is known. To address $\degp_{C, T}'$, note that for each community node $u\in N_C$, $\degp'(u)$ may be represented similarly:
\[
\degp'(u) = \degp_C'(u) + \degp_B'(u)
\]
This reduces the problem of specifying $\degp'(u)$ to specifying $\degp_C'(u)$ and $\degp_B'(u)$. Since the community node subnetwork has already been generated, we set $\degp_C'(u) \gets d_C^o(u)$. Next, recalling that $\degp_B(u) := \frac{\degp_{B, T}}{\degp_T}\degp(u)$, we make the specification $\degp_B'(u) := \frac{\degp_{B, T}}{\degp_T'}\degp(u)$ (which must be solved for via $\phi'_T$, in the following). So, in total, we have
\[
\degp'(u) = \begin{cases}d_C^o(u) + \frac{\degp_{B, T}}{\degp_T'}\degp(u),&u\in N_C\\
\degp(u),&u\in N_B\end{cases}
\]
Therefore we can solve for $\degp_T'$ with the equation
\begin{align*}
\degp_T' :=& \;\sum_{u\in N_C\cup N_B} \degp'(u)\\
=& \;\sum_{u\in N_C}\left[d_C^o(u) + \frac{\degp_{B, T}}{\degp_T'}\degp(u)\right] + \sum_{u\in N_B}\degp(u)\\
=& \;d_{C, T}^o + \frac{\degp_{B, T}}{\degp_T'}\degp_{C, T} + \degp_{B, T}
\end{align*}
Where $d_{C, T}^o:=\sum_{u\in N_C}d_C^o(u)$. The solution for $\degp_T'$ from this quadratic is
\begin{equation}
\label{eq:doT}
\degp_T' = \dfrac{\degp_{B, T} + d_{C, T}^o}{2} + \sqrt{\dfrac{(\degp_{B, T} + d_{C, T}^o)^2}{4} + \degp_{C, T}\degp_{B, T}}
\end{equation}
which then immediately gives the full vector ${\bf \degp}'$. We can now simulate the remaining edges in the network. Specifically, for each $u\in N_B$ and each $v\in N_C\cup N_B$, we simulate an edge according to
\begin{equation}\label{eq:ccm_append}
\prob\left(A_{uv} = 1\right) = \frac{\degp'(u)\degp'(v)}{\degp_T'}\text{ independent across node pairs}
\end{equation}
We solve for $\strp'$ analogously. Then for each $u\in N_B$ and each $v\in N$, we simulate an edge weight according to
\[
W_{uv} = \begin{cases}f_{uv}({\degp'},{\strp'})\xi_{uv}, & A_{uv} = 1\\
0, & A_{uv} = 1
\end{cases}
\]
where $\xi{\sim}F$, is as it was for the generation of the community node subnetwork.

The above simulation steps correspond precisely to the continuous configuration model with parameters $(\degp',\strp',F, \sigma)$. Some basic computational trials have shown that, for large networks, the solution for $\phi_T'$ is quite close to $\phi_T$. Therefore, for each $u\in N_B$, $\expv(d(u))$ is almost exactly $\degp(u)$, i.e.\ what it would be under the model in \ref{comm_model}, without background nodes. The same holds for the strengths and expected strengths. Together with equation \ref{eq:ccm_append}, this implies the background nodes are behaving according to the continuous configuration model, even as they are a sub-network within a larger network with communities. 

To illustrate these points, we simulated a sample network from the default framework with parameters $n = 5,000$, $n_b = 1, 000$, $s_e = s_w = 3$, disjoint communities, and other parameters specified by \ref{settings}. These settings are akin to what was used in subsection \ref{Simulations} of the main text. First we plotted $\phi'$ and $\psi'$ against the empirical strengths and degrees with lowess curves to check the match. Figure $\ref{fig:str_deg}$ shows the fit is essentially linear. 
\begin{figure}[!htb]
\centering
\makebox{
\includegraphics[scale = 0.20]{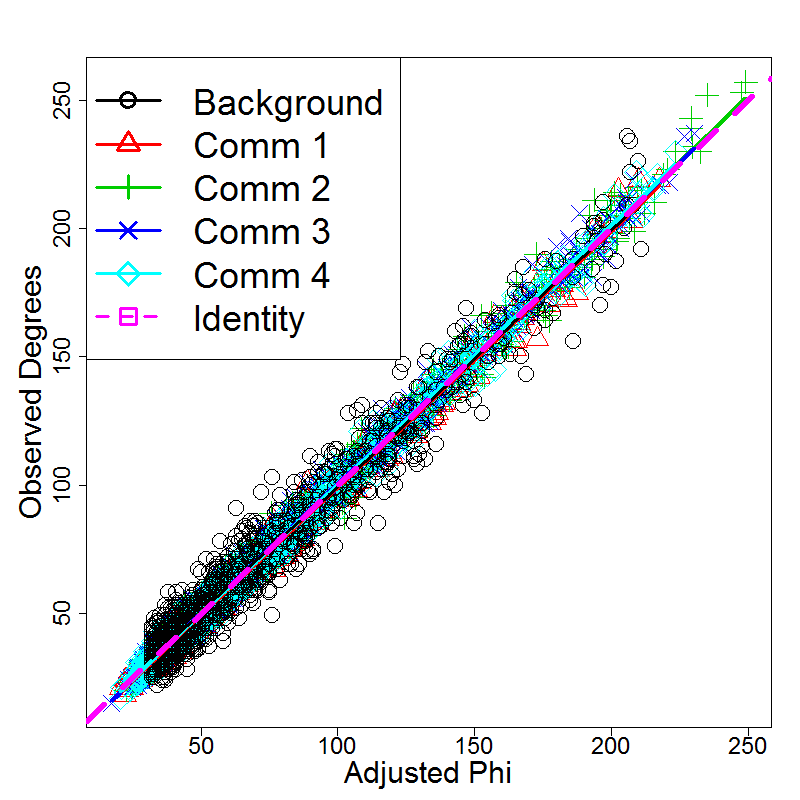}
\includegraphics[scale = 0.20]{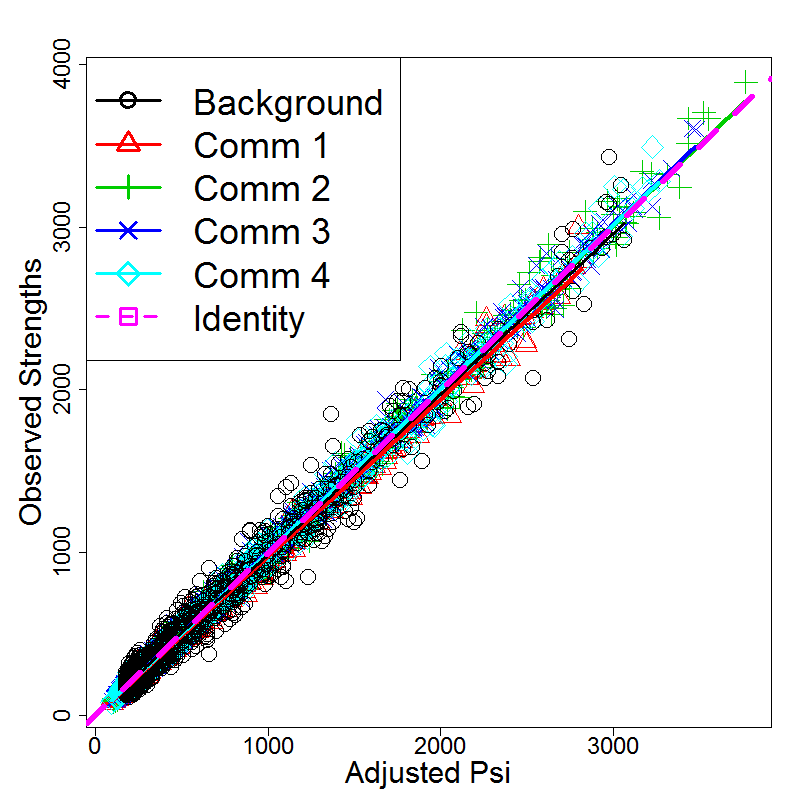}
}
\caption{\label{fig:str_deg} Empirical degrees/strengths vs.\ adjusted parameters for the example network}
\end{figure}
Second, for each node $u\in N$ and for each node block $B$ (either a true community or the background node set) we may calculate an empirical $z$-score for $S(u,B,\mathcal{G})$, as described in subsection \ref{Method:CLT} of the main text. The $z$-score for $S(u,B,\Net)$ is a measure of connection significance, with respect to the continuous configuration model (and also modularity, see Section \ref{theory-wm}) between $u$ and $B$. Let $K$ be the number of true communities in the network. For each $i,j = 1,\ldots,K+1$, where $K+1$ is the index of the background node block, we computed the empirical average of $z$-statistics between nodes $u$ from node block $i$ the node block $B$ corresponding to index $j$. Theses empirical averages can be arranged in a $(K+1)\times(K+1)$ matrix showing the average inter-block connectivities of the network. In Figure \ref{fig:avg_stats} we display a visualization of this matrix, which shows preferential connection within communities, and roughly null connection between the background nodes and all blocks.

\begin{figure}[!htb]
\centering
\includegraphics[scale = 0.6]{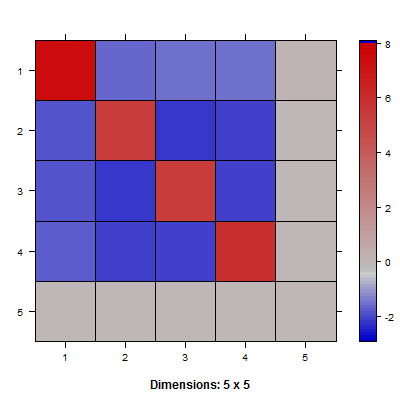}
\caption{\label{fig:avg_stats} Average empirical $z$-statistics between nodes and node blocks}
\end{figure}